\documentclass[12pt]{article}
\usepackage[utf8]{inputenc}
\usepackage[margin=1in]{geometry}
\usepackage{setspace}
\usepackage{amsmath}
\usepackage{amssymb}

\allowdisplaybreaks
\usepackage{hyperref}
\makeatletter
\newcommand{\printfnsymbol}[1]{%
  \textsuperscript{\@fnsymbol{#1}}%
}
\makeatother
\usepackage{graphicx}
\usepackage{caption}
\usepackage{natbib}
\usepackage{subcaption}
\usepackage{color}
\usepackage{bbm}
\usepackage{float}

\usepackage{amsthm}
\usepackage[inline,shortlabels]{enumitem}
\newtheorem*{remark}{Remark}
\usepackage[vlined, ruled]{algorithm2e}
\usepackage{tikz}
\newtheorem{assumption}{Assumption}
\usepackage{dsfont}
\usepackage{mathrsfs}
\usepackage{lmodern}

\usetikzlibrary{arrows.meta}

\renewcommand{\Pr}{\mathrm{Pr}}

\newcommand{\Zdagger}{\mathcal{Z^{\dagger}}}
\newcommand{\bdagger}{b^{\dagger}}
\newcommand{\pidagger}{\pi^{\dagger}}

\newcommand{\Br}{\mathrm{Br}}
\newcommand{\B}{\mathcal{B}}
\newcommand{\T}{T(\xi^{\bar{u}})}

\newcommand{\im}{\mathrm{Im}}

\newcommand{\ptest}{\gamma_i}
\newcommand{\pc}{\gamma_i}

\newcommand{\U}[2]{U_{#1}(#2)}

\newcommand{\Lu}{\max_{i' \in I, a \in A} u_{i'}(a)}
\newcommand{\La}{\max_{i' \in I} \{|A_{i'}| \cdot |A_{-i'}|\}}

\newcommand{\deltaT}{\delta_T}
\newcommand{\rhat}{\hat{r}}
\usepackage{xcolor}
\newtheorem{definition}{Definition}
\newtheorem{proposition}{Proposition}
\newtheorem{lemma}{Lemma}
\newtheorem{theorem}{Theorem}
\newtheorem{corollary}{Corollary}

\usepackage{appendix}
\DeclareMathOperator*{\argmax}{argmax} 
\DeclareMathOperator{\E}{\mathbb{E}}

\title{Learning with Episodic Hypothesis Testing in General Games: A Framework for Equilibrium Selection}
\author{Ruifan Yang and Manxi Wu\thanks{Ruifan Yang is with the School of Operations Research and Information Engineering at Cornell University. Manxi Wu is with the department of Civil and Environmental Engineering at UC Berkeley. Contact: manxiwu@berkeley.edu}
}
\date{This version: July 30, 2025}

\begin{document}
\maketitle
\abstract{
We introduce a new hypothesis testing-based learning dynamics in which players update their strategies by combining hypothesis testing with utility-driven exploration. In this dynamics, each player forms beliefs about opponents’ strategies and episodically tests these beliefs using empirical observations. Beliefs are resampled either when the hypothesis test is rejected or through exploration, where the probability of exploration decreases with the player's (transformed) utility. In general finite normal-form games, we show that the learning process converges to a set of approximate Nash equilibria and, more importantly, to a refinement that selects equilibria maximizing the minimum (transformed) utility across all players. Our result establishes convergence to equilibrium in general finite games and reveals a novel mechanism for equilibrium selection induced by the structure of the learning dynamics. }

\section{Introduction}
When games admit multiple equilibria, a central and pressing question is: which equilibria are likely to emerge from the learning process? This challenge has led to research on adaptive dynamics as a mechanism for equilibrium refinement. In this paper, we propose a hypothesis testing-based learning
dynamics for general finite normal-form games that enable equilibrium refinement based on players’ exploration behavior.

Our work builds on an extensive body of research investigating dynamic models in which players revise their strategies in response to observed actions or realized payoffs. In structured games such as zero-sum, dominance-solvable, or potential games, belief-based models like fictitious play are known to converge to Nash equilibria under best-response dynamics \citep{Fudenberg1998}. In more general settings, Bayesian learning and hypothesis-testing based learning dynamics converge to Nash equilibria under the grain-of-truth assumption \citep{Kalai1993, Foster2003, Nachbar1997, Nachbar2001, Nachbar2005}, while calibrated forecasting ensures convergence to correlated equilibria when players best respond to predictions \citep{Foster1997, Kakade2004, Mannor2007}.  In repeated two-player games, \citet{Jindani2022} extends hypothesis testing based learning to achieve Pareto-efficient equilibria via exploration. Other payoff-based dynamics, such as regret minimization leads to coarse correlated equilibria \citep{Hart2000, Germano2007, Foster2006}. Stochastic learning rules, including log-linear learning and trial-and-error learning, introduce random perturbations into the strategy update process and are analyzed using Markov chain methods to characterize their stochastically stable states — the set of states visited with high probability as the randomness vanishes \citep{Young2009, Marden2009a, Pradelski2012, MardenShamma2012, Borowski2019}. When pure equilibria exist, these dynamics show that the learning process not only can reach equilibria but also select the welfare-maximizing ones.

Existing equilibrium refinement results focus on learning dynamics that lead to welfare-maximizing or Pareto-efficient equilibria, and often rely on strong assumptions, such as potential game structure \citep{MardenShamma2012}, the existence of pure Nash equilibria \cite{Pradelski2012}, or restriction to two-player games \citep{Jindani2022}. There is limited understanding on how the design of a learning algorithm influences which equilibrium is selected, and how players might steer the learning outcome by tuning parameters in their own dynamics. Our work addresses this gap by analyzing a hypothesis-testing based learning dynamic in which each player's exploration behavior — modulated by their sensitivity to utility —  governs the long-run equilibrium outcome of the learning dynamics. 


In our learning dynamic, each player maintains a belief about others’ strategies, plays a smooth best response, and periodically tests whether the belief aligns with observed actions. 
A key feature of this dynamics is the integration of episodic hypothesis testing with utility-sensitive exploration. Hypothesis testing serves to detect whether a player’s belief is statistically inconsistent with observed empirical play. If the discrepancy exceeds a tolerance threshold with high confidence, the belief is rejected and resampled. Exploration provides a complementary mechanism for belief revision, which allows players to experiment when dissatisfied with their utility, even if their belief passes the statistical test. The probability of such exploration decreases with expected utility, mapped through a transformation function that encodes the player's sensitivity to dissatisfaction. This transformed utility plays a central role in determining how likely a player is to explore at a given utility level and, in turn, influences long-run equilibrium selection.

We analyze the long-run behavior of this learning dynamics by characterizing its stochastically stable states. Formally, the joint evolution of players' beliefs and strategies induces a Markov chain over a finite state space, where each state represents a tuple of player beliefs and their corresponding smooth best response strategies. A state is stochastically stable if it retains positive stationary probability distribution as the exploration rate becomes very small. In other words, when the exploration rate is positive but close to zero, the system spends most of its time in the stochastically stable states.

We show that the stochastically stable set lies within the set of approximate Nash equilibria. Specifically, it selects those equilibria that maximize the minimum transformed utility across all players (Theorem \ref{theorem:main}). There are two key contributions in this result. First, this result guarantees that the stochastically stable set is composed of approximate equilibria in general finite games, which implies that this learning dynamics will eventually play those approximate equilibria for the vast majority of time. Second, this result identifies a principled refinement criterion among equilibria: those with higher minimum transformed utility are more stable, as players with lower transformed utility are more prone to explore and destabilize the equilibrium.

The intuition behind this result follows directly from the design of the learning dynamics. States in which players hold beliefs that are inconsistent with observed actions are unlikely to persist, due to the high probability of hypothesis rejection. We show in Proposition \ref{lemma: sigma, tau, M} that all consistent states correspond to $\epsilon$-Nash equilibria under appropriate conditions on the best response function smoothness, hypothesis testing tolerance, and belief discretization parameters. Hence, all long-run stochastically stable outcomes must be approximate equilibria. At each consistent state (i.e., approximate Nash equilibria), the player with the lowest transformed utility (i.e. the least satisfied) has the highest probability to initiate exploration and deviate. As a result, equilibria that maximize the minimum transformed utility are least likely to be destabilized and thus become the most stable in the long run. We prove this result using stochastic stability analysis and resistance-tree methods \cite{Young1993}, and show that such equilibria correspond to recurrent classes of the Markov chain with the minimal stochastic potential.

Furthermore, our analysis reveals how the structure of the utility transformation functions governs equilibrium refinement. When all players use identical transformation functions, the learning dynamics select equilibria that maximize the minimum utility across agents, leading to a max-min refinement. In contrast, when players can choose distinct transformation functions, the dynamics may favor a particular player whose transformation consistently maps their utility to values lower than those of others. Because such a player explores more frequently at equilibrium, the stochastically stable set selects equilibria that maximize this player’s utility. Corollary \ref{cor:twocases} formally characterizes this effect and shows how the choice of transformation functions directly governs which players are favored by the equilibrium refinement process.

\section{Model and Preliminaries}\label{sec: model}
\subsection{The Static Game Model}
We consider a static game \( G = (I, A = (A_i)_{i \in I}, (u_i)_{i \in I}) \), where:
\begin{itemize}
    \item \( I \) is a finite set of players, with \( |I| = n \),
    \item \( A_i \) is a finite set of actions for player \( i \), and \( A := \prod_{i \in I} A_i \) denotes the set of joint action profiles. Let \( A_{-i} := \prod_{j \ne i} A_j \) denote the set of joint actions of players other than \( i \),
    \item \( u_i : A \to \mathbb{R} \) is the payoff function for player \( i \).
\end{itemize}

For each player \( i \), a mixed strategy \( \pi_i = (\pi_i(a_i))_{a_i \in A_i} \in \Delta_i \) is a probability distribution over \( A_i \), where \( \Delta_i \) denotes the simplex on \( A_i \). A joint strategy profile is \( \pi = (\pi_i)_{i \in I} \in \Delta := \prod_{i \in I} \Delta_i \), where $\Delta$ is the space of joint mixed strategies. The expected utility for player \( i \) given \( \pi \) is:
\[
U_i(\pi) = \mathbb{E}_{a \sim \pi}[u_i(a)].
\]
Each player \( i \in I \) holds a \emph{belief} \( b_i = (b_{ij})_{j \in I \setminus \{i\}} \) of their opponents' strategies, where each \( b_{ij} \in \Delta_j^M \) represents player \( i \)'s belief about player \( j \)'s strategy. We assume that beliefs are discretized in that each \( b_{ij} \) is in a \emph{discretized probability simplex} over \( A_j \) with granularity parameter \( M \in \mathbb{N}_+ \) defined as follows: 
\[
b_{ij}  \in \Delta_j^M := \left\{ b_{ij} \in \mathbb{R}_{\geq 0}^{|A_j|} \;\middle|\; \sum_{a_j \in A_j} b_{ij}(a_j) = 1,\; b_{ij}(a_j) \in \left\{ \frac{m}{M} : m = 0, \ldots, M \right\} \right\},~ \forall j \in I \setminus \{i\}, ~ \forall i \in I.  
\]
Accordingly, player \( i \)'s \emph{discretized belief space} is \(\mathcal{B}_i := \prod_{j \ne i} \Delta_j^M
\). We denote the joint belief vector as $b = (b_i)_{i \in I}$ and the \emph{joint discretized belief space} as
\(
\mathcal{B} := \prod_{i \in I} \mathcal{B}_i.
\)


We next define the smooth best response function and \(\epsilon\)-Nash equilibrium.

\begin{definition}[Smooth Best Response]
For any player \( i \in I \) and temperature parameter \( \sigma > 0 \), the smooth best response function \( \Br_i^\sigma : \mathcal{B}_i \to \Delta_i \) is given by:
\begin{align}\label{eq:smoothed_br}
    \Br_i^\sigma(a_i \mid b_i) = \frac{e^{\frac{1}{\sigma}U_i(a_i, b_{i})}}{\sum_{a_i' \in A_i} e^{\frac{1}{\sigma}U_i(a_i', b_{i})}}, \quad \forall a_i \in A_i, \quad \forall b_i \in \mathcal{B}_i,
\end{align} 
where $U_i(a_i', b_{i}) = \sum_{a_{-i}\in A_{-i}} u_i(a_i, a_{-i})b_{i}(a_{-i})$. The image set of the smooth best response function is $\im\left(\Br^{\sigma}_{i}\right):= \left\{\pi_{i} \left\vert \pi_i = \Br^\sigma_{i}(b_i), ~ b_i \in \B_i\right.\right\}$. 
\end{definition}
We denote the joint smooth best response function as 
$\Br^\sigma(b) = (\Br_i^\sigma(b_i))_{i \in I}$, and the joint smooth best response function for all players other than $i$ as 
$\Br^\sigma_{-i}(b_{-i}) = (\Br_j^\sigma(b_j))_{j \in I\setminus\{i\}}$.

\begin{definition}[\(\epsilon\)-Nash Equilibrium]
For any \( \epsilon > 0 \), a strategy profile \( \pi^* \in \Delta \) is an \(\epsilon\)-Nash equilibrium if:
\[
U_i(\pi_i^*, \pi_{-i}^*) \ge U_i(\pi_i', \pi_{-i}^*) - \epsilon, \quad \forall \pi_i' \in \Delta_i, \quad \forall i \in I.
\]
As \( \epsilon \to 0 \), \( \pi^* \) becomes a Nash equilibrium.
\end{definition}

\subsection{Preliminaries for Hypothesis Testing} \label{subsec:testing}

A hypothesis test evaluates whether player \( i \)'s belief \( b_i \) is close to the true strategy profile $\pi_{-i}$ of the opponents. Given a tolerance level $\tau > 0$, we define:

\begin{itemize}
    \item[-] The \emph{null hypothesis} \( H_0 \): The strategy profile $\pi_{-i}$ is within $\tau$-distance to player \( i \)'s belief $b_i$:
    \[ \|\pi_{-i} - b_i\|_2 \leq \tau .\]
    \item[-] The \emph{alternative hypothesis} \( H_1 \): The strategy profile $\pi_{-i}$ is not within $\tau$-distance to player \( i \)'s belief $b_i$:
    \[ \|\pi_{-i} - b_i\|_2 > \tau .\]
\end{itemize}

For each player $i$, let $\{a_{-i}^t\}_{t=1}^{T}$ be $T$ independent and identically distributed (i.i.d.) action profiles sampled from $\pi_{-i}$.
    For a given significance level $\alpha \in (0,1)$, each player $i$ conducts a hypothesis test that rejects $H_0$ if and only if 
    \begin{align*}
        \|\hat{\pi}_{-i} - b_i \|_2 > \tau  + \sqrt{\frac{|A_{-i}|}{2T}\cdot\ln\bigl(\frac{2}{\bar\alpha}\bigr)}, 
    \end{align*}
    where $\hat{\pi}_{-i}$ is the empirical  estimate of the opponents' strategy profile:
\[
  \hat{\pi}_{-i}(a_{-i}) =
  \frac{1}{T}\sum_{t=1}^{T}\mathbf1\{a_{-i}^{t} = a_{-i}\}, 
  \qquad \forall a_{-i}\in A_{-i}.
\]
That is, the reject region for player $i$'s hypothesis test is given by
\begin{align}\label{eq:rejection_region}
R_i(\alpha) = \left\{(a^t_{-i})_{t = 1}^T| \left\vert \|\hat{\pi}_{-i} - b_i\| > \tau + \sqrt{\frac{|A_{-i}|}{2T} \cdot \ln \left(\frac{2}{\alpha}\right)}\right.\right\}.    
\end{align}
The type I and type II errors of the hypothesis test are defined as follows: 
\begin{itemize}
    \item[-] \emph{Type I error (false positive)}: rejecting \( H_0 \) even though $\|\pi_{-i}-b_i\|_2\leq\tau$.
    \item[-] \emph{Type II error (false negative)}: failing to reject \( H_0 \) even though $\|\pi_{-i}-b_i\|_2>\tau$.
\end{itemize}
\begin{proposition} \label{prop: hypothesis test}
For any $\alpha \in (0,1)$, if 
\begin{align}\label{eq:T_bound}
    T \geq  T(\alpha) : = \max\left\{\left.\frac{2|A_{-i}|}{(\|\pi_{-i} - b_i\|_2-\tau)^2}\,\ln\left(\frac{2}{\alpha}\right) \right\vert i \in I, b_i \in \mathcal{B}_i, \pi_{-i}\in \im(\Br^{\sigma}_{-i}),\|\pi_{-i}-b_i\|_2 > \tau\right\},
\end{align} 
then the maximum type I and type II errors of the hypothesis test are upper bounded by $\alpha$ for all players, all beliefs and the associated smooth best response strategies. That is, for any player $i$, any belief $b_i \in \B_i$ and smooth best response strategy profile $\pi_{-i} = \Br_i^{\sigma}(b_{-i})$, 
\begin{align*}
& \Pr(\text{reject } H_0| ~   \|b_i - \pi_{-i}\|_2 \leq \tau) \leq \alpha, \quad  \Pr(\text{fail to reject } H_0| ~\|b_i - \pi_{-i}\|_2 > \tau) \leq \alpha. 
\end{align*}
\end{proposition}
Proposition~\ref{prop: hypothesis test} ensures that, if the number of samples $T$ exceeds a certain threshold, then both the Type I and Type II error rates of the hypothesis test are uniformly bounded by the significance level $\alpha$, across all players, all discretized beliefs, and all smooth best response strategy profiles. The proof of this result is provided in Appendix ~\ref{appendix: proof of hypothesis test}.

\begin{remark}
The hypothesis test introduced above is a simple test based on an $\ell_2$-distance threshold and a Chernoff-type bound on the empirical estimation error. It is not necessarily the most powerful or optimal test for detecting discrepancies between beliefs and strategies. Since the focus of this paper is not on optimal test design but on the learning dynamics and convergence behavior in general games, we adopt this test for simplicity. Our convergence results remain valid under any hypothesis testing procedure that ensures small Type I and Type II error given a finite number of samples.
\end{remark}

\section{Learning Dynamics and Stochastically Stable Set}\label{sec:algo}

In this section, we present the learning dynamics (Algorithm~\ref{alg:learning with hypothesis testing}), where players periodically test their beliefs about other players' strategies and resample new beliefs of opponents based on hypothesis test result and their utility. We also present our main theorem that demonstrates the long-run outcomes of the learning dynamics. 
\subsection{Learning Dynamics}
We begin by outlining the overall structure of the learning dynamics, before detailing the key quantities that govern the algorithm’s dynamics.
In Algorithm~\ref{alg:learning with hypothesis testing}, each player $i$ begins with an initial belief $b_i^0$ about others’ strategies and selects a smooth best response strategy $\pi_i^0$ corresponding to their belief, as in \eqref{eq:smoothed_br}. The learning proceeds in \emph{epochs} $k=1,2,\ldots$. In each epoch $k$, each player $i$ maintains belief $b_i^k$ and plays the game according to a smooth best response strategy $\pi_i^k= \Br_i^{\sigma}(\pi_i^k)$ for $T(\xi^{\bar{u}})$ rounds, collecting observations of joint action profiles. Here, $\xi \in (0,1)$ is a hyperparameter of the learning dynamics and $\bar{u}$ is a sufficiently large constant, which will be described in details later. 
Each player $i$ independently decides whether to conduct a hypothesis test against their current belief $b_i^k$, with probability $\gamma_i > 0$. The hypothesis test is as defined in Section \ref{subsec:testing}, with $\xi^{\bar{u}}$ being the significance level and $T(\xi^{\bar{u}})$ being the sample size. If the null hypothesis is rejected (indicating a significant discrepancy between the observed actions and belief $b_i^k$ with high probability), player $i$ samples a new belief $b_i^{k+1}$ and updates their strategy to the corresponding best response $\pi_i^{k+1} = \mathrm{Br}_i^\sigma(b_i^{k+1})$. If the null hypothesis is not rejected, player $i$ may still sample a new belief (referred to as exploration) with probability $\xi^{f_i(U_i(\pi^k_i,b^k_i))}$, where $f_i(\cdot)$ is an increasing function and $U_i(\pi^k_i,b^k_i)$ is player $i$'s anticipated utility given their belief and best response. If neither testing nor exploration occurs, the player does not update their belief and strategy. Fig. \ref{fig:learning} illustrates the belief update process for each player in each decision epoch based on randomized testing, hypothesis rejection, and exploration decisions. 
\begin{figure}[ht]
    \centering
    \includegraphics[width=0.45\textwidth]{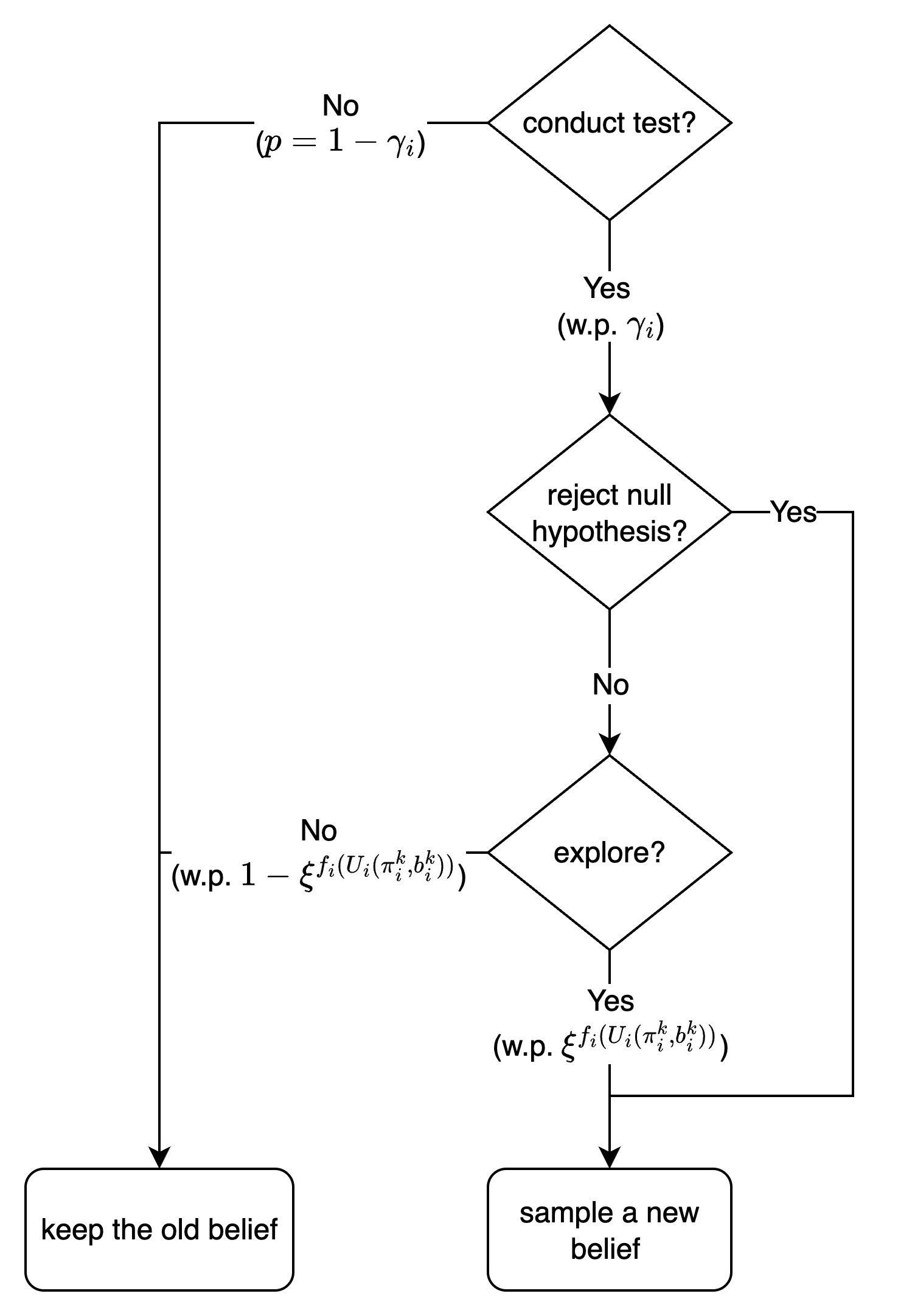}  
    \caption{Belief update flowchart in each epoch $k$.}
    \label{fig:learning}
\end{figure}

We now provide more details on the key quantities of the learning algorithm: 
\begin{itemize}
    \item $\gamma_i>0$ is the probability that each player conducts a test in each epoch.
    \item $\xi \in (0,1)$ is a hyperparameter that affects hypothesis test significance level, epoch length, and players' probability of exploration.
    \item $f_i: \mathbb{R} \to \mathbb{R}_{>0}$ is an increasing function that maps player $i$'s utility to a positive value. When player $i$'s hypothesis test fails to reject in epoch $k$, the player explores (resamples a new belief) with probability $\xi^{f_i(U_i(\pi_i^k, b_i^k))}$, where $U_i(\pi_i^k, b_i^k)$ is the expected utility of player $i$ given their belief and smoothed best response. Since $\xi \in (0,1)$ and $f_i(\cdot)$ is increasing, the exploration probability is higher when the utility $U_i(\pi_i^k, b_i^k)$ is low. The function $f_i(\cdot)$ can be viewed as player $i$'s sensitivity function, governing how their exploration probability changes with their utility, and may differ for different players. We refer to $f_i(U_i(\pi_i^k, b_i^k))$ as player $i$'s transformed utility in epoch $k$. 
     \item $\bar{u} > \max_{\pi \in \Delta} \sum_{i \in I} f_i(U_i(\pi))$ is a constant larger than the maximum total transformed utility of all players in the game. The constant $\bar{u}$ being sufficiently large ensures that the test significance level $\xi^{\bar{u}}$ in each epoch is sufficiently small.   
    \item $T(\xi^{\bar{u}})$ is the length of each epoch, which is also the sample size of each hypothesis test. Following from Proposition \ref{prop: hypothesis test}, $T(\xi^{\bar{u}})$ satisfies \eqref{eq:T_bound}, and ensures that the hypothesis test conducted by each player has Type I and Type II errors upper bounded by $\xi^{\bar{u}}$.
    \item $\psi_i = \left( \psi_i(b'_i \mid b_i) \right)_{b_i, b'_i \in \mathcal{B}_i}$ is player $i$'s belief resampling probability, where $\psi_i(b'_i \mid b_i)$ is the probability that player $i$ resamples belief $b'_i$ for the next epoch, given their current belief $b_i$. This resampling rule is quite general: it is player-specific and allows the new belief distribution to depend on the current belief. We impose the following assumption to ensure sufficient exploration: the resampling probability has full support over the belief space, i.e.,
$$
\psi_i(b'_i \mid b_i) \geq \lambda \quad \forall b_i, b'_i \in \mathcal{B}_i,
$$
where $\lambda>0$. 

\end{itemize}

\begin{algorithm}[ht]
\SetKwInOut{Input}{Input}
\SetKwInOut{Init}{Initalization}
\SetKwIF{ProbIf}{ProbElse}{probif}{with probability}{:}{}{else}{}
\Input{$\xi \in (0,1)$, $\ptest (0,1)$ for each $i \in I$\;
$f_i: \mathbb{R} \to \mathbb{R}_+$ is an increasing utility transformation function for each player $i \in I$\;
$\T$ satisfies \eqref{eq:T_bound} and $\bar{u} > \max_{\pi \in \Delta} \sum_{i \in I} f_i(\U{i}{\pi}) $\;
$\psi_i = \left( \psi_i(b'_i \mid b_i) \right)_{b_i, b'_i \in \mathcal{B}_i}$ is the belief resampling probability for each player $i \in I$, $\psi_i(b'_i \mid b_i)>\lambda>0$ for all $b_i, b'_i \in \mathcal{B}_i$ and all $i \in I$\;
}
\Init{Each player $i$ holds a belief $b^0_i \in \mathcal{B}_i$ and chooses a smooth best response strategy $\pi^0_i = \Br^{\sigma}_i(b_i^0)$.}
\For{epoch $k = 1,2,\cdots$}{
    Each player $i \in i$ plays the game with strategy $\pi^k_i$ for $\T$ periods: $(a^t_i) \sim \pi^k_i$ for ${t = 1,\cdots, \T}$, and observe full action profiles $(a^t)_{t = 1,\cdots , \T}$\;
    \For{each player $i = 1,\cdots, |I|$}{
        \eProbIf{$\ptest$}
        {Conduct a hypothesis test for $b^k_i$ using observations $(a^t_{-i})_{t = 1,\cdots , \T}$\;
        \eIf{ $(a^t_{-i})_{t = 1}^{\T} \in R_i(\xi^{\bar{u}})$ as in \eqref{eq:rejection_region}}{
            Sample $b^{k+1}_i\in \mathcal{B}_i$ according to $\psi_i(\cdot|b_i^k)$ and 
             $\pi^{k+1}_i = \Br^{\sigma}_i(b^{k+1}_i)$\;
        }{
            \eProbIf{$\xi^{f_i(\U{i}{\pi^k_i,b^k_i})}$}{
            Sample $b^{k+1}_i\in \mathcal{B}_i$ according to $\psi_i(\cdot|b_i^k)$ and 
            $\pi^{k+1}_i = \Br^{\sigma}_i(b^{k+1}_i)$\;}
            {$b^{k+1}_i = b^k_i$; $\pi^{k+1}_i = \pi^k_i$\;}
        }
        }{$b^{k+1}_i = b^k_i$; $\pi^{k+1}_i = \pi^k_i$\;}
    }
}
\caption{Learning with Episodic Hypothesis Testing}
\label{alg:learning with hypothesis testing}
\end{algorithm}


A key feature of the learning dynamics is the integration of hypothesis testing with utility-sensitive exploration. The hypothesis test evaluates whether player $i$'s belief $b_i^k$ is within $\tau$ distance to the actual strategy profile $\pi^k_{-i}$. Specifically, when the test rejects the null hypothesis with significance level $\xi^{\bar{u}}$, it indicates with probability at least $1-\xi^{\bar{u}}$, the discrepancy $\|\pi^k_{-i} - b_i^k\|_2 > \tau$; in this case, player will resample a new belief. Conversely, when the null hypothesis fails to be rejected, it indicates that $b_i^k$ is $\tau$-consistent with $\pi^k_{-i}$ with probability at least $1-\xi^{\bar{u}}$. In this scenario, belief resampling occurs when the player explores with probability $\xi^{f_i(U_i(\pi_i^k, b_i^k))}$, where $f_i(\cdot)$ encodes the player's sensitivity to utility dissatisfaction. The integration of statistical hypothesis testing and utility-sensitive exploration ensures that the learning process corrects inconsistent beliefs with high probability while allowing adaptive exploration when current utility is unsatisfactory.

\subsection{Stochastically Stable Set Characterization}
In the learning dynamics, we refer to the tuple \((b, \pi)\) as the \emph{state} of the system. We note that for any belief \( b_i \in \mathcal{B}_i \), each player \( i \in I \) always selects their strategy \( \pi_i \in \Delta_i \) via the smooth best response function \( \pi_i = \Br_i^\sigma(b_i)\) with temperature parameter \( \sigma \). Therefore, given any $\xi > 0$, the learning dynamics induces a Markov chain on the finite state space: 
\[\mathcal{Z}:= \{(b, \pi) \mid b_i \in \mathcal{B}_i,\, \pi_i = \Br^{\sigma}_i(b_i), ~i \in I\}.\] 

\begin{definition}[Consistent state]\label{def:consistent_state}
A state \( z = (b, \pi) \in \mathcal{Z} \) is \textit{consistent} if, for all players \( i \in I \), their belief \( b_i \) is within \(\tau\)-distance of the strategy profile \(\pi_{-i}\), and \(\pi\) is a smooth best response to \(b\). That is, the set of all consistent states is given by 
\[\Zdagger := \{z = (b, \pi) \in \mathcal{Z} \mid \|b_i - \pi_{-i}\|_2 \leq \tau, \quad \forall i \in I\}.\]
\end{definition}

\begin{assumption} \label{assump: sigma, tau, M}
Given any $\epsilon>0$, we assume that parameters $\sigma, \tau, M$ satisfy: 
\begin{subequations}
\begin{align}
    \sigma &\leq \frac{\epsilon}{2\cdot\log(\max_{i \in I}|A_i|)}, \quad  \tau \leq \frac{\epsilon \cdot \sigma}{2 \sqrt{|A|} \cdot (\max_{i \in I, a \in A} u_i(a) \cdot (\max_{i \in I} |A_i| \cdot |A_{-i}|)} \label{eq: sigma and tau},\\
       M &\geq \frac{|I| \cdot \max_{i \in I} |A_{i}| }{\tau} \cdot \left( 1 + \frac{\sqrt{\max_{i \in I} |A_{i}|}}{\sigma} \cdot \max_{i \in I, a \in A} u_i(a) \cdot \max_{i \in I} |A_i||A_{-i}| \cdot |I| \right).  \label{eq: granularity M}
\end{align} 
\end{subequations}
\end{assumption}

Assumption~\ref{assump: sigma, tau, M} ensures that the smooth best response temperature parameter $\sigma$ and the hypothesis testing tolerance level $\tau$ are sufficiently small, while the belief granularity parameter $M$ is sufficiently large. A small $\sigma$ guarantees that the smooth best response closely approximates the exact best response to any belief. A small $\tau$ ensures that, under consistency, each player's belief is sufficiently close to the actual strategy of their opponents. A large $M$ ensures that for any strategy profile, there exists a discretized belief vector that satisfies the consistency condition. The following proposition shows that under Assumption \ref{assump: sigma, tau, M}, each consistent state $z \in \Zdagger$ is an $\epsilon$-Nash equilibrium. The proof of this result is provided in Appendix ~\ref{appendix: large M}.
\begin{proposition}\label{lemma: sigma, tau, M}
    For any $\epsilon > 0$, under Assumption \ref{assump: sigma, tau, M},
    the consistent state set $\Zdagger$ is non-empty. Furthermore, for any $z^{\dagger} = (b^{\dagger}, \pi^{\dagger}) \in \Zdagger$, $\pi^{\dagger}$ is an $\epsilon$-Nash equilibrium.
\end{proposition}

For each epoch $k$, we denote $z^k = (b^k, \pi^k)$ as the state in epoch $k$. Given $\xi$, the state transition probability matrix is $P^{\xi}= (P^{\xi}_{zz'})_{zz' \in \mathcal{Z}}$, where $P^{\xi}_{zz'}$ is the probability of state transitioning from $z = (b,\pi)$ to $z' = (b', \pi')$ in one epoch:
\[ P^{\xi}_{zz'} := \mathbb{P}(z^{k+1} = z' \mid z^k = z), \quad \forall z, z' \in \mathcal{Z}. \]
The following lemma shows that the state transition Markov chain is finite and has a unique stationary distribution. 

\begin{lemma}\label{lemma: P_ep}
    The Markov process described by transition $P^\xi$ has a finite state space and is aperiodic and irreducible for all $\xi \in (0, 1)$, and has a unique stationary distribution $\mu^{\xi}$.
\end{lemma}
A state is stochastically stable if it is in the support set of the stationary distribution $\mu^{\xi}(z)$ as $\xi\to 0$. 

\begin{definition}[\cite{Young1993}]
    A state $z \in \mathcal{Z}$ is \textit{stochastically stable} 
    if \[\lim_{\xi \to 0} \mu_z^\xi > 0,\]
    where $\mu_z^{\xi}$ is the stationary distribution of state $z$ given $P^{\xi}$.
\end{definition}

We denote the set of stochastically stable states by $\mathcal{Z}^*$. These are the states of the belief-based learning dynamics that retain positive probability in the limit stationary distribution as the perturbation parameter $\xi \to 0$. The notion of stochastic stability captures the asymptotic state distribution of the Markov process induced by the learning dynamics under vanishing $\xi$. Formally, for any $\delta > 0$, there exists $\bar{\xi} > 0$ such that for all $\xi \in (0, \bar{\xi})$, the stationary distribution $\mu^\xi$ places at least $1 - \delta$ probability mass on the stochastically stable set $\mathcal{Z}^*$. This implies that when $\xi$ is sufficiently small, the learning dynamics visits states in $\mathcal{Z}^*$ for the majority (more than $1-\delta$ fraction) of epochs. Therefore, states in the stochastically stable set can be viewed as \emph{high-probability long-run outcomes} of the learning dynamics—that is, the beliefs and strategies most frequently observed in the evolution of the system.

 Before presenting the theorem, we introduce the following assumption, which ensures that for each player 
$i$, regardless of the opponents’ beliefs and strategies, there exists a belief of player 
$i$ that is inconsistent with the opponents’ strategies, and the corresponding smooth best response of player $
i$ is also inconsistent with the opponents’ beliefs. 


\begin{assumption}\label{assump: bad strategy}
        For any player $i \in I$, any belief profile $b_{-i} \in \prod_{j \neq i} \mathcal{B}_j$ and strategy profile $\pi_{-i} = \Br^\sigma_{-i}(b_{-i})$, there exists $\tilde{b}_i \in \Delta^M_{-i}$, and $\tilde{\pi}_i = \Br^{\sigma}_i(\tilde{b}_i)$ such that belief $\tilde{b}_i$ is at least $\tau$ distance away from strategy $\pi_{-i}$:
        \begin{align*}
            \|\tilde{b}_i - \pi_{-i}\|_2 > \tau,
        \end{align*}
        and strategy $\tilde{\pi}_i$ is at least $\tau$ distance away from belief $b_{ji}$ for any $j \neq i$:
        \begin{align*}
            \|\tilde{\pi}_i -  b_{ji}\|_2 > \tau, \quad \forall j \neq i.
        \end{align*}
\end{assumption}
This assumption ensures that, regardless of the system state, each player $i$ has a belief $\tilde{b}_i$ that can, with high probability, trigger belief resampling for all players. Specifically, player $i$ will resample their own belief with probability greater than $1 - \xi^{\bar{u}}$, because $\tilde{b}_i$ is at least $\tau$-distant from the opponents’ strategy profile $\pi_{-i}$. Furthermore, the corresponding smooth best response $\tilde{\pi}_i$ is inconsistent with the opponents’ beliefs about player $i$, which can also prompt the opponents to resample. As $\xi \to 0$, the probability of belief resampling triggered by $\tilde{b}_i$ approaches 1. In Appendix \ref{appendix: bad strategy}, we show in Lemma \ref{lemma:simple_condition} that this assumption is mild and can be satisfied when the image of each player’s smooth best response function contains more than $|I|$ strategies and the tolerance parameter $\tau$ is sufficiently small.



\begin{theorem} \label{theorem:main}
Suppose Assumptions \ref{assump: sigma, tau, M} and \ref{assump: bad strategy} hold. The stochastically stable state set is given by 
\[
\mathcal{Z}^{*} = \left\{ z =(b, \pi) \in \Zdagger \left\vert 
\min_{i \in I} f_i(U_i(\pi_i, b_i)) = \max_{z' = (b', \pi') \in \Zdagger} \min_{i \in I} f_i(U_i(\pi'_i, b'_i))\right.\right\}.
\]
\end{theorem}


Theorem~\ref{theorem:main} shows that any stochastically stable state must be a consistent state, as defined in Definition~\ref{def:consistent_state}. Following Proposition~\ref{lemma: sigma, tau, M}, the strategy profile associated with every consistent state is an \(\epsilon\)-Nash equilibrium for sufficiently small \(\tau\), \(\sigma\), and large \(M\). Hence, the long-run outcomes of the learning process lie within the \(\epsilon\)-Nash equilibrium set. This property arises from the hypothesis test, which rejects states where empirical play deviates significantly from held beliefs with high probability, thereby preventing inconsistent belief-strategy tuples from attaining stochastic stability.

Furthermore, the long-run outcome of the learning dynamics selects among consistent states (i.e., approximate equilibria) those that maximize the minimum transformed utility across players. This effect arises from the structure of the exploration mechanism. The probability of leaving a consistent state depends on the likelihood that some player initiates exploration. The dominant contribution to this probability comes from the player with the highest exploration probability, \(\xi^{f_i(U_i(\pi_i, b_i))}\). Since \(\xi < 1\), this is governed by the player with the lowest transformed utility \(f_i(U_i(\pi_i, b_i))\), who is most inclined to explore. As a result, the equilibrium strategy profiles that maximize \(\min_{i \in I} f_i(U_i(\pi_i, b_i))\) among all approximate equilibria are the most stable and remain in the stochastically stable set as \(\xi \to 0\).

\begin{corollary}\label{cor:twocases}
Suppose Assumptions \ref{assump: sigma, tau, M} and \ref{assump: bad strategy} hold. 
\begin{itemize}
    \item[(i)] If \( f_i(u) = f_j(u) \) for any \( i, j \in I \) and any $u \in [\underline{u}, \bar{u}]$, where $\underline{u}$ (resp. $\bar{u}$) is the minimum (resp. maximum) utility of all players given all feasible $(\pi, b)$, then 
    \[
\mathcal{Z}^{*} = \left\{ z \in \Zdagger \left\vert 
\min_{i \in I} U_i(\pi_i, b_i) = \max_{z' = (b', \pi') \in \Zdagger} \min_{i \in I} U_i(\pi'_i, b'_i)\right.\right\}.
\]
\item If there exists a player $\hat{i}$ such that \( f_{\hat{i}}(u_{\hat{i}}) \leq  f_j(u_j) \) for all \( j \in I \setminus \{i\} \), $u_{\hat{i}} \in [\underline{u}_{\hat{i}}, \bar{u}_{\hat{i}}]$ and $u_{j} \in [\underline{u}_j, \bar{u}_j]$, where $\underline{u}_{i}$ (resp. $\bar{u}_{i}$) is the minimum (resp. maximum) utility of each player $i \in I$ given all feasible $(\pi_i, b_i)$, then 
\[
\mathcal{Z}^{*} = \left\{ z \in \Zdagger \left\vert 
U_{\hat{i}}(\pi_{\hat{i}}, b_{\hat{i}}) = \max_{z' = (b', \pi') \in \Zdagger} U_{\hat{i}}(\pi'_{\hat{i}}, b'_{\hat{i}})\right.\right\}.
\]
\end{itemize}
\end{corollary}

Corollary \ref{cor:twocases} shows that the utility transformation functions \( \{f_i(\cdot)\}_{i \in I} \) govern how equilibrium refinement may favor certain players over others by modulating their exploration probabilities. When \( f_i \) is identical across all players, the learning dynamics select equilibria that maximize the minimum utility across players, corresponding to a max-min equilibrium refinement (Corollary \ref{cor:twocases} (i)). More generally, the choice of \( \{f_i(\cdot)\}_{i \in I} \) influences which player is most likely to initiate exploration at equilibrium. A player $\hat{i}$ whose function \( f_{\hat{i}}(\cdot) \) consistently maps their utility to a lower value than those of others across equilibria will have the highest exploration probability at all equilibrium states. As a result, the equilibrium refinement process selects equilibria that maximize the utility of player $\hat{i}$, since such equilibria are least likely to be destabilized by their exploration (Corollary \ref{cor:twocases} (ii)). 

\paragraph{Example $1$ (Cooperative Outcome):} 
Two players simultaneously choose between Stag ($S$) and Hare ($H$).  The payoff matrix is
\[
\begin{array}{|c|c|c|}
\hline
 \text{Player 1}\backslash\text{Player 2}   & S & H \\ \hline
S & (4,4)  & (0,3)  \\ \hline
H & (3,0)  & (3,3)  \\ \hline
\end{array}.
\]

The Nash Equilibria are:
\begin{itemize}
  \item Pure Nash: \((S,S)\), \((H,H)\)
  \item Mixed Nash: each player plays \(S\) with probability \ \(\tfrac{3}{4}\), and \(H\) with probability\ \(\tfrac{1}{4}\). 
\end{itemize}
Suppose each player $i$ chooses identity utility transformation function $f_i(u) = u$. In the mixed NE and the pure NE $(H,H)$, the minimum utility of both players is $3$; whereas in $(S,S)$, the minimum utility of both players is $4$. Therefore, for sufficiently small $\epsilon$ and $\xi$, the learning plays the fully cooperative and utility maximizing outcome $(S, S)$ for the majority of the time.

\paragraph{Example $2$ (Transformation Function Governs Convergence):}

Two players simultaneously choose between  Opera (\(O\)) and Football (\(F\)).  The payoff matrix is
\[
\begin{array}{|c|c|c|}
\hline
\text{Player 1}\backslash\text{Player 2} & O      & F        \\ \hline
O   & (2,\,1)    & (0,\,0)    \\  \hline
F   & (0,\,0)    & (1,\,2)    \\ \hline
\end{array}.
\]
The Nash Equilibria are:
\begin{itemize}
  \item Pure Nash: \((O,O)\), \((F,F)\)
  \item Mixed Nash: player $1$ plays \(O\) with probability \(\tfrac{2}{3}\) and  \(F\) with probability\ \(\tfrac{1}{3}\); player $2$ plays \(O\) with probability\ \(\tfrac{1}{3}\), and \(F\) with probability\ \(\tfrac{2}{3}\). 
\end{itemize}
Suppose each player $i$ chooses the identity utility transformation function $f_i(u) = u$. In the mixed NE, the minimum utility of the two players is $\tfrac{2}{3}$; whereas in $(O,O)$ and $(F,F)$, the minimum utility of both players is $1$. Therefore, for sufficiently small $\epsilon$ and $\xi$, learning players strategy profiles $(O,O)$ and $(F,F)$ for most of the time. The mixed Nash equilibrium is not in the stochastically stable set. 

On the other hand, suppose that player $1$ chooses the identity utility transformation function $f_1(u) = u$, and player $2$ chooses $f_2(u) = u-0.1$. Then, in $(O,O)$, the minimum utility of both players is $0.9$, whereas in $(F,F)$, the minimum utility of both players is $1$. In this case, for sufficiently small $\epsilon$, the only stochastically stable state is $(F,F)$, where player $2$ receives the higher utility of $2$ by increasing their exploration probability.

\section{Proof of Theorem \ref{theorem:main}} \label{section: proof of main theorem}
To prove Theorem~\ref{theorem:main}, we proceed in three steps. First, we show that the joint evolution of beliefs and strategies forms a regular perturbation of an unperturbed Markov chain with $\xi=0$, where hypothesis tests have no error and exploration occurs with probability tending to zero. Second, we show that the recurrent classes of the unperturbed chain are exactly the consistent states, and every inconsistent state transitions to one of these consistent states with positive probability. Finally, we apply the resistance-tree method from \cite{Young1993} to compute the stochastic potential of each recurrent class and identify the stochastically stable set as the set of recurrent classes with minimum stochastic potential in a regular perturbed Markov chain. We compute the resistance of transitions between recurrent classes in our learning dynamics and show that the classes minimizing this potential correspond exactly to the set of consistent states that maximize the minimum transformed utility across players.

\subsection{Regular Perturbation}
We define the limiting process as the exploration parameter $\xi \to 0$. In this limit, the original Markov process defined by the transition matrix $P^{\xi} = (P^{\xi}_{zz'})_{z,z' \in \mathcal{Z}}$ converges to an unperturbed process governed by the transition matrix $P^0 = (P^0_{zz'})_{z,z' \in \mathcal{Z}}$. In the unperturbed process, both Type I and Type II error rates in hypothesis testing become $0$, and players do not explore when they fail to reject the null hypothesis. As a result, for each player $i \in I$, if $\|b_i - \pi_{-i}\|_2 \leq \tau$, then they retain their current belief with probability $1$. Conversely, if $\|b_i - \pi_{-i}\|_2 > \tau$, then the player rejects the null hypothesis and updates their belief with probability $1$.

We introduce the definition of a regular perturbation and show that for any $\xi>0$, the original Markov process $P^{\xi}$ is a regular perturbation of the unperturbed process $P^0$.


\begin{definition}[\citet{Young1993}]
Let $P^0$ be the transition matrix of a stationary Markov chain defined on a finite state space $\mathcal{Z}$. Let $P^{\xi}$ be the transition matrix of a family of Markov chains on $\mathcal{Z}$ perturbed from $P^0$, indexed by $\xi \in (0,\bar{\xi})$ for some $\bar{\xi}>0$. The family $\{P^{\xi}\}_{\xi > 0}$ is a \emph{regular perturbation} of $P^0$ if the following conditions hold for all $z, z' \in \mathcal{Z}$:
\begin{enumerate}[(i)]
    \item $P^{\xi}$ is aperiodic and irreducible for all $\xi \in (0,\bar{\xi})$.
    \item $\lim_{\xi \to 0} P^{\xi}_{zz'} = P^0_{zz'}$.
    \item If $P^{\xi}_{zz'} > 0$ for some $\xi$, then there exists $r_{zz'} \geq 0$ such that $0 < \lim_{\xi \to 0} P^{\xi}_{zz'}\cdot \xi^{-r_{zz'}} < \infty$.
\end{enumerate}
\end{definition}

\begin{lemma} \label{lemma: regular perturbation}
\begin{enumerate}[label=(\arabic*), itemsep=1em]

\item For any $z= (b, \pi), z'= (b', \pi') \in \mathcal{Z}$, the transition probability of the unperturbed process $P^0$ satisfies:
\[
P^0_{zz'} = \prod_{i \in I} P^{0}_{i, zz'},
\]
where $P^{0}_{i, zz'}$ denotes the probability of player $i$ updating their belief from $b_i$ to $b'_i$:
\begin{align} \label{eq: P0}
P^{0}_{i, zz'} &= 
\begin{cases}
    1, & \text{if } \|b_i - \pi_{-i}\|_2 \leq \tau \text{ and } b_i = b'_i, \\
    0, & \text{if } \|b_i - \pi_{-i}\|_2 \leq \tau \text{ and } b_i \neq b'_i, \\
    (1-\pc) + \pc \cdot \psi_i(b'_i|b_i), & \text{if } \|b_i - \pi_{-i}\|_2 > \tau \text{ and } b_i = b'_i, \\
    \pc \cdot \psi_i(b'_i|b_i), & \text{if } \|b_i - \pi_{-i}\|_2 > \tau \text{ and } b_i \neq b'_i.
\end{cases}
\end{align}

\item The Markov process $P^{\xi}$ is a regular perturbation of $P^0$. Moreover, for any $z = (b,\pi), z' = (b', \pi') \in \mathcal{Z}$,
\[
0 < \lim_{\xi \to 0} P^{\xi}_{zz'}\cdot \xi^{-r_{zz'}} < \infty,
\]
where
\begin{align}\label{eq:rzz}
r_{zz'} = \sum_{\substack{i \in \{i \in I ~|~ b_i \neq b'_i,\\ \|b_i - \pi_{-i}\|_2 \leq \tau\}}} f_i(\U{i}{\pi_i,b_i}).
\end{align}
\end{enumerate}
\end{lemma}

In Lemma \ref{lemma: regular perturbation}, part (1) formalizes the structure of transitions in the unperturbed process: players update beliefs only when the hypothesis test rejects the null, and belief updates are resampled over the discretized belief space. Part (2) verifies that the belief-based learning dynamics define a \emph{regular perturbation} of the limiting Markov process $P^0$ as the exploration parameter $\xi \to 0$, and shows how  every nonzero transition in the perturbed process changes with $\xi$. This result is essential for applying the resistance-tree framework of stochastic stability. The proof of Lemma \ref{lemma: regular perturbation} is in Appendix \ref{apx:regular_perturbation}. 

\subsection{Recurrent Communication Class of the Unperturbed Process}

A key concept of Markov chains is the recurrent communication classes. Once the chain enters a recurrent communication class, it never leaves, and any state in the class can be reached from any other state in the class with positive probability. This is formally defined as follows:


\begin{definition}
    A recurrent communication class of a Markov chain is a non-empty subset of states \( C \subseteq \mathcal{Z} \) such that:
    \begin{enumerate}[(i)]
        \item For all \( z, w \in C \), there exist integers \( k, k' \geq 0 \) such that
        \[
        \Pr(z^k = w \mid z^0 = z) > 0, \quad \text{and} \quad \Pr(z^{k'} = z \mid z^0 = w) > 0.
        \]
        That is, all states in \( C \) communicate with each other.
    
        \item For any \( z \in C \) and any \( w \notin C \), we have
        \[
      \Pr(z^k = w \mid z^0 = z) = 0, \quad \forall k \geq 0.
        \]
        That is, the class \( C \) is \emph{closed}: the Markov chain cannot exit once it enters.
    \end{enumerate}
    \end{definition}

In the following lemma, we show that each consistent belief state $z \in \Zdagger$ forms a recurrent communication class of the unperturbed Markov process $P^0$. 

\begin{lemma}\label{lemma:recurrent_class}
    For each consistent belief state $z \in \Zdagger$, $\{z\}$ is a recurrent communication class of the unperturbed process $P_0$. Furthermore, there are no other recurrent communication class of $P_0$. 
\end{lemma}

\begin{proof}
    We first show that every state $z$ in $\Zdagger$ is an absorbing state. For any state $z = (b , \pi) \in \Zdagger$, every player's hypothesis is consistent with the actual strategy, i.e. \(\|b_i - \pi_{-i}\|_2 \leq \tau\) for all $i \in I$. 
    From Lemma \ref{lemma: regular perturbation}, we know for all $i \in I$, when $\|b_i - \pi_{-i}\|_2 \leq \tau$, we have
    \begin{align*}
        P^0_{i, zz'} &= 
    \begin{cases}
        1, \quad &\text{ if } b_i = b'_i, \\
        0, \quad &\text{ if } b_i \neq b'_i.
    \end{cases}
    \end{align*} 
    Therefore, all players will keep their current belief and strategy, and the Markov chain remains in the same state:
    \[P^0_{zz'} = \prod_{i \in I} P^0_{i, zz'}  = \begin{cases}
        1, \quad &\text{ if } z = z', \\
        0, \quad &\text{ if } z \neq z'.
    \end{cases}\]
    Therefore, any consistent belief state $z \in \Zdagger$ is an absorbing state and $\{z\}$ is a singleton recurrent communication class of the unperturbed process $P^0$ for all $z \in \Zdagger$.

    To show that there are no other recurrent communication class in the unperturbed process, it suffices to show that for any non-consistent state $w = (b^w, \pi^w)\not \in \Zdagger$, there exists a belief consistent state $z = (b^z, \pi^z) \in \Zdagger$ that is reachable from $w$. 
    Since $w$ is not a consistent belief state, there exists at least one player $i$ whose belief is not consistent, i.e.
    \begin{align}\label{eq: inconsistent at w}
        \|b^w_i - \pi^w_{-i}\|_2 > \tau.
    \end{align} 
    From Assumption \ref{assump: bad strategy}, we know that there exists a belief $\tilde{b}_i$ that satisfies: 
    \begin{align}\label{eq: bad belief for z}
        \|\tilde{b}_i - \pi^w_{-i}\|_2 > \tau,
    \end{align}
    and the associated smooth best respond strategy $\tilde{\pi}_i = \Br^{\sigma}_i(\tilde{b}_i)$ is not consistent with any other player's belief in $w$:
    \begin{align}\label{eq: bad belief for w}
        \|\tilde{\pi}_i - b^w_{ji}\|_2 > \tau, \quad \forall j \in I\setminus \{i\}.
    \end{align} 
We consider the state $\tilde{w} = (\tilde{b}_i, b^w_{-i},\tilde{\pi}_i, \pi^w_{-i})$. We show that the state transition path $w \to \tilde{w}  \to z$ has positive probability in the unperturbed process.
    
    \paragraph{Transition step $1$: $w \to \tilde{w}$.} Player $i$ updates their belief from $b_i^w$ to $\tilde{b}_i$ with probability $P^0_{i,w\tilde{w}}$, and all other players $j \neq i$ do not update. From \eqref{eq: P0} and \eqref{eq: inconsistent at w}, $P^0_{i,w\tilde{w}}$ is given by:
\begin{align*}
        P^0_{i,w\tilde{w}} &= \begin{cases}
            (1-\pc) + \pc \cdot \phi_i(\tilde{b}_i|b_i^w), \quad &\text{ if } b^w_i = \tilde{b}_i,\\
            \pc \cdot \phi_i(\tilde{b}_i|b_i^w), \quad &\text{ if } b^w_i \neq \tilde{b}_i,
        \end{cases}\\
        &\geq \pc \cdot \lambda.
    \end{align*}
    For any $j \neq i$, player $j$'s belief $b^w_j$ does not change when the state transits from $w$ to $\tilde{w}$:
 \begin{align*}
        P^0_{j,w\tilde{w}} &= \begin{cases}
            1, \quad &\text{ if } \|b^w_j - \pi^w_{-j}\|_2 \leq \tau,\\
            (1-\gamma_j) + \gamma_j \cdot \psi_j(b_j^w|b_j^w), \quad &\text{ if } \|b^w_j - \pi^w_{-j}\|_2 > \tau 
        \end{cases}\\
        & \geq (1-\gamma_j) + \gamma_j \cdot \lambda.
    \end{align*}
    Hence, the probability of state transits from $w$ to $\tilde{w}$ is positive:
 \begin{align*}
        P^0_{w\tilde{w}} &= P^0_{i,w\tilde{w}}\cdot \prod_{j \neq i} P^0_{j,w \tilde{w}}\geq (\pc \cdot \lambda)\prod_{j \neq i} ((1-\gamma_j) + \gamma_j \lambda)) > 0.
\end{align*}\paragraph{Transition step $2$: $\tilde{w} \to z$.} Player $i$ updates their belief from $\tilde{b}_i$ to $b^z_i$, and every other player $j \neq i$ updates their belief from $b^w_j$ to $b^z_j$.
    Since $\|\pi^{\tilde{w}}_{-j} - b^{\tilde{w}}_{j}\|_2 \geq \|\tilde{\pi}_i - b^w_{ji}\|_2 \stackrel{\eqref{eq: bad belief for w}}{>} \tau$, from \eqref{eq: P0}, $P^0_{j,\tilde{w}z}$ satisfies
\begin{align*}
        P^0_{j,\tilde{w}z} &= \begin{cases}
            (1-\gamma_j) + \gamma_j \cdot \phi_j(b_j^z | b_j^w), \quad &\text{ if } b^w_j = b^z_j,\\
            \gamma_j \cdot \phi_j(b_j^z | b_j^w),\quad &\text{ if } b^w_j \neq b^z_j,
        \end{cases}\\
        &\geq \gamma_j \cdot \lambda.
    \end{align*}
    From \eqref{eq: bad belief for z}, we have $\|\tilde{b}_i - \pi^w_{-i}\|_2 > \tau$. From \eqref{eq: P0}, 
 \begin{align*}
        P^0_{i,\tilde{w}z} &= \begin{cases}
    (1-\pc) + \pc \cdot \psi_i(b_i^z|\tilde{b}_i), \quad &\text{ if }  \tilde{b}_i = b^z_i, \\
    \pc  \cdot \psi_i(b_i^z|\tilde{b}_i), \quad &\text{ if } \tilde{b}_i \neq b^z_i.
    \end{cases}\\
        & \geq  \pc \cdot \lambda.
    \end{align*}
    Hence, transition from $\tilde{w}$ to $\tilde{z}$ happens with positive probability:
\begin{align*}
        P^0_{\tilde{w}\tilde{z}} &= P^0_{i,\tilde{w}\tilde{z}}\cdot \prod_{j \neq i} P^0_{j,\tilde{w}\tilde{z}} \geq \prod_{i \in I}(\pc \cdot \lambda) > 0.
    \end{align*}
Summarizing, the probability of the state transitioning from $w$ to $z$ in two epoches is positive: 
    \begin{align*}
        \Pr(z^2 = z \mid z_0 = w) \geq P^0_{w\tilde{w}} \cdot  P^0_{\tilde{w}z}  > 0. 
    \end{align*}
    Hence, any non belief consistent state $w \not \in \Zdagger$ cannot be in a recurrent communication class of the unperturbed Markov chain.
\end{proof}


\subsection{Tree Resistance and Stochastic Potential}
Having established in Steps 1 and 2 that the learning dynamics form a regular perturbation of an unperturbed Markov chain and that the recurrent classes of the unperturbed chain correspond to consistent states, we now characterize which of these classes comprise the stochastically stable set. While all consistent states are absorbing in the unperturbed dynamics, vanishing but nonzero perturbations induce transitions between them with small probabilities. To compare the relative stability of these recurrent classes under perturbation, we adopt the notion of resistance introduced in \cite{Young1993}, which captures the leading-order exponent of the transition probability as a function of the perturbation magnitude. Formally, we present the definitions of edge and path resistance, $z$-tree and stochastic potential as follows: 

\begin{definition}[Edge resistance, Path Resistance, Resistence Tree, and Stochastic Potential {\citep{Young1993}}]\label{def:resistance}
Consider a perturbed Markov process on a finite state space \(\mathcal{Z}\), with transition matrix \(P^\xi\). Let $G$ be a fully connected directed graph with states $\mathcal{Z}$ being the node set. 

\begin{itemize}
    \item \textbf{Edge Resistance.} For any \(z, z' \in \mathcal{Z}\), if \(P^{\xi}_{zz'} > 0\), the resistance of the edge \((z \to z')\) is the unique number \(r_{zz'} \geq 0\) such that
    \[
    0 < \lim_{\xi \to 0} P^{\xi}_{zz'} \cdot \xi^{-r_{zz'}} < \infty.
    \]
    If \(P^{\xi}_{zz'} = 0\), we define \(r_{zz'} = \infty\).
    \item \textbf{Path Resistance.} For a transition path \(\rho_{z_1 \to z_n} = z_1 \to z_2 \to \cdots \to z_n\), the resistance of the path is the sum of its edge resistances:
    \[
    R(\rho_{z_1 \to z_n}) = \sum_{j=1}^{n-1} r_{z_j z_{j+1}}.
    \]
    \item \textbf{Minimum Path Resistance.} For any \(z, z' \in \mathcal{Z}\), let $\hat{r}_{wz}$ be the smallest path resistance over all paths that start in $w$ and end in $z$:
    \[\hat{r}_{wz} = \min_{\rho_{w \to z}} R(\rho_{w \to z}).\]
\end{itemize}
Let $\mathcal{G}$ be a fully connected directed graph with each node being a recurrent class of the unperturbed Markov chain. The edge weight between any two recurrent classes (nodes) $C_i, C_j \in \mathcal{C}$ is defined as the minimum path resistance $\hat{r}_{ij}$ across all paths that begin in $C_i$ and end in $C_j$ in the original graph $G$. The following concepts are defined on the graph $\mathcal{G}$: 

\begin{itemize}
    \item \textbf{\(j\)-Tree.} For any node \(C_j \in \mathcal{C}\), a \(j\)-tree $\Gamma$ is a spanning directed tree over \(\mathcal{C}\) such that for every node \(C_i \neq C_j\), there is a unique directed path from \(C_i\) to \(C_j\). The weight of a \(j\)-tree is the sum of the weight of its constituent edges.

    \item \textbf{Stochastic Potential.} The stochastic potential \(\phi(j)\) of a recurrent class \( C_j \in \mathcal{C}\) is defined as the minimum weight over all possible \(j\)-trees:
    \[
    \phi(j) = \min_{\Gamma \text{ is a } j\text{-tree}} \sum_{(C_j \to C_{j'}) \in \Gamma} \rhat_{j j'}.
    \]
\end{itemize}
\end{definition}

In Definition~\ref{def:resistance}, the edge resistance $r_{zz'}$ characterizes the leading-order exponent of the one-step transition probability $P^\xi_{zz'}$ as the perturbation parameter $\xi \to 0$, i.e., $P^\xi_{zz'} = \Theta(\xi^{r_{zz'}})$. Transitions with higher resistance become exponentially less likely as $\xi$ decreases. The resistance of a path is defined as the sum of the resistances along its constituent edges. The minimum path resistance between two states $z$ and $z'$ is the smallest total resistance among all finite-length paths from $z$ to $z'$. This quantity determines the leading-order exponent of the total probability that the process transitions from $z$ to $z'$ over a finite number of epochs, i.e., $\mathbb{P}(z \to z') = \Theta(\xi^{\hat{r}_{zz'}})$. In particular, the minimum path resistance directly reflects how fast the probability of reaching $z'$ from $z$ decays as the perturbation vanishes. Edge resistance, path resistance, and minimum path resistance are all defined on the graph $G$, whose nodes are the states $z \in \mathcal{Z}$. These quantities characterize the exponential rate at which the probability of transitioning between states decays as the perturbation parameter $\xi \to 0$, either in a single epoch (edge resistance) or over any finite sequence of epochs (path and minimum path resistance).


On the other hand, the graph $\mathcal{G}$ can be viewed as a reduction of $G$, with nodes representing recurrent classes of the unperturbed Markov chain ($\xi = 0$). The probability of state transition from one recurrent class to another is zero in the unperturbed chain, but becomes positive when $\xi > 0$. The edge weight $\hat{r}_{ij}$, defined as the minimum path resistance from a state in $C_i$ to a state in $C_j$, characterizes the leading-order exponent of the transition probability from $C_i$ to $C_j$, i.e., $\mathbb{P}(C_i \to C_j) = \Theta(\xi^{\hat{r}_{ij}})$. An edge with higher weight implies a faster exponential decay of this probability as $\xi \to 0$. The notion of a $j$-tree captures how recurrent class $C_j$ can be reached from all other recurrent classes through a directed spanning tree. 
The stochastic potential $\phi(j)$, defined as the weight of the minimum $j$-tree, quantifies the leading-order exponent of the maximum probability that the system transitions from all other recurrent classes into $C_j$. 
Indeed, the following lemma from \citet{Young1993} formalizes this characterization: the stochastically stable states are those contained in recurrent classes with the minimum stochastic potential.

\begin{lemma}[\citet{Young1993}] \label{theorem: Young 1993}
Let $P^0$ be a stationary Markov process on the finite state space $\mathcal{Z}$ with recurrent communication classes $C_1, \ldots, C_J$. Let $P^\xi$ be a regular perturbation of $P^0$.
Then, a state $z$ is stochastically stable if and only if $z$ is contained in a recurrent communication class $C_j$ that minimizes its stochastic potential.
\end{lemma}

From Lemma \ref{lemma:recurrent_class}, we know that the set of all recurrent classes of our unperturbed Markov chain is the set of consistent states, i.e. $C= \{\{z\}\}_{z \in \mathcal{Z}^{\dagger}}$. 
To complete the proof of Theorem~\ref{theorem:main}, it remains to identify which of the recurrent class (i.e. consistent state) is associated with the minimum stochastic potential. 
The following lemma characterizes the edge and minimum path resistances between any pair of consistent states, which will be used for computing the weight of any $z$-tree.

\begin{figure}[ht]
\centering
\begin{subfigure}{0.45\textwidth}
\centering
\begin{tikzpicture}[
    inner/.style={circle, draw, fill=blue!20, minimum size=4mm,inner sep=0pt,font=\bfseries},
    outer/.style={circle, draw, minimum size=4mm,inner sep=0pt}
  ]
\tikzset{>={Latex[width=1mm,length=1mm]}}

\begin{scope}[shift={(0,2.8cm)}]
  \foreach \i in {1,...,4} {
    \node[inner] (I\i) at ({90+(\i-1)*360/4}:0.7cm) {\i};
  }
\end{scope}

  \foreach \i in {1,...,10} {
  \pgfmathtruncatemacro{\labelval}{\i + 4}
  \node[outer] (O\i) at ({90+(\i-1)*360/10}:1.3cm) {\labelval};
}

  \foreach \i in {1,...,10} {
    \foreach \j in {1,...,10} {
      \ifnum\i<\j
        \draw[gray!70,<->] (O\i) -- (O\j);
      \fi
    }
  }

  \foreach \i in {1,...,10} {
    \foreach \j in {1,...,4} {
      \draw[gray!70,->] (O\i) -- (I\j);
    }
  }
\foreach \i in {1,...,4} {
    \foreach \j in {1,...,4} {
    \ifnum\i<\j
      \draw[gray!70,->] (I\i) -- (I\j);
    \fi
    }
  }

\end{tikzpicture}
\caption{Graph $G$}
\end{subfigure}
\hfill
\begin{subfigure}{0.45\textwidth}
\centering
\begin{tikzpicture}[
    inner/.style={circle, draw, fill=blue!20, minimum size=5mm,inner sep=0pt},
    outer/.style={circle, draw, minimum size=4mm,inner sep=0pt}
  ]
  \tikzset{>={Latex[width=2mm,length=2mm]}}
  \foreach \i in {1,...,4} {
    \node[inner] (I\i) at ({90+(\i-1)*360/4}:2cm) {\i};
  }

  \foreach \i in {1,...,4} {
    \foreach \j in {1,...,4} {
      \ifnum\i<\j
        \draw[<->] (I\i) -- (I\j);
      \fi
    }
  }

\end{tikzpicture}
\caption{Graph $\mathcal{G}$}
\end{subfigure}
\caption{In, graph $G$, the blue nodes represent consistent states (i.e. recurrent classes of the unperturbed Markov chain), and white nodes represent inconsistent states. Graph $\mathcal{G}$ is the reduced graph where each node is a consistent state, and each edge weight is the minimum path resistance of the connecting states in $G$.}
\end{figure}
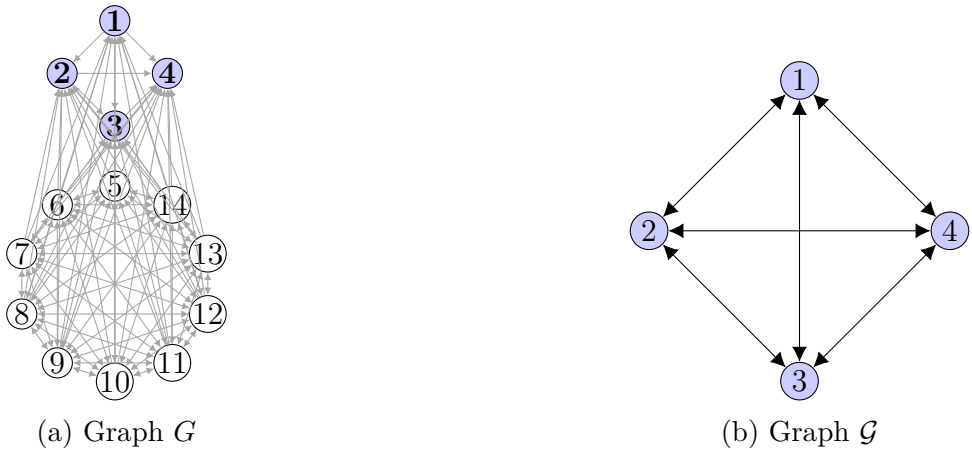



\begin{lemma}[Minimum resistance between consistent states] \label{lemma: out edge resistance}
Let \( z = (b,\pi),\, z' = (b',\pi') \in \mathcal{Z} \). The edge resistance between \( z \) and \( z' \) in $G$ is

\begin{align}\label{eq:edge_resistence}
r_{zz'} \;=\; \sum_{i\in I} r_{i,zz'},\qquad 
r_{i,zz'} \;=\;
\begin{cases}
f_i\bigl(U_i(\pi_i, b_i)\bigr), & \text{if } \|b_i - \pi_{-i} \|_2 \le \tau \text{ and } b_i \ne b'_i, \\[4pt]
0, & \text{otherwise}.
\end{cases}
\end{align}
    For any $z = (b^z,\pi^z) \in Z^\dagger$ and any $w \in Z^\dagger$, the edge weight $\hat{r}_{zw}$ of graph $\mathcal{G}$ is 
    \begin{align}\label{eq:path_resistence}
        \hat{r}_{zw} = \min_{i \in I} f_i(\U{i}{\pi^z_i, b^z_i}).
    \end{align}

\end{lemma}

\medskip 
\noindent\emph{Proof of Lemma \ref{lemma: out edge resistance}.} 
From part $(2)$ of Lemma \ref{lemma: regular perturbation}, for any $z = (b,\pi), z' \in \mathcal{Z}$, we have that 
\[
0 < \lim_{\xi \to 0} P^{\xi}_{zz'}\cdot \xi^{-r_{zz'}} < \infty,
\]
for
\begin{align*}
r_{zz'} = \sum_{\substack{i \in \{i \in I ~|~ b_i \neq b'_i,\\ \|b_i - \pi_{-i}\|_2 \leq \tau\}}} f_i(\U{i}{\pi_i,b_i}).
\end{align*}
Therefore, the edge resistance between \( z \) and \( z' \) in graph $G$ satisfies
\begin{align*}
r_{zz'} \;=\; \sum_{i\in I} r_{i,zz'},\qquad 
r_{i,zz'} \;=\;
\begin{cases}
f_i\bigl(U_i(\pi_i, b_i)\bigr), & \text{if } \|b_i - \pi_{-i} \|_2 \le \tau \text{ and } b_i \ne b'_i, \\[4pt]
0, & \text{otherwise}.
\end{cases}
\end{align*}

We next prove that the edge weight between any two consistent state in $\mathcal{G}$ is given by \eqref{eq:path_resistence}. Consider any belief consistent states $z = (b^z, \pi^z), w = (b^w, \pi^w) \in \Zdagger$. We first prove that $\hat{r}_{zw} \leq \min_{i \in I} f_i(\U{i}{\pi_i, b_i})$. Recall that the edge weight $\hat{r}_{zw}$ is defined as the minimum resistance over all possible paths between $z$ and $w$. It suffices to present a transition path from $z$ to $w$ with resistance 
    $\min_{i \in I} f_i(\U{i}{\pi^z_i, b^z_i})$. 
    Let player $i^{\dagger}$ be a player with the smallest $f_i(\U{i}{\pi^z_i, b^z_i})$ among all players. That is, \[f_{i^{\dagger}}(\U{i^{\dagger}}{\pi^z_{i^{\dagger}}, b^z_{i^{\dagger}}}) = \min_{j \in I} f(\U{j}{\pi^z_j, b^z_j}).\]
    From Assumption \ref{assump: bad strategy}, there exists a belief $\tilde{b}_{i^{\dagger}} \in \mathcal{B}_{i^{\dagger}}$ for player $i^{\dagger}$ that is at least $\tau$ distance away from strategy $\pi^z_{-i^{\dagger}}$:
    \begin{align*}
        \|\tilde{b}_{i^{\dagger}} - \pi^z_{-i^{\dagger}}\|_2 > \tau,
    \end{align*}
    and its $\sigma$-smooth strategy $\tilde{\pi}_{i^{\dagger}} = \Br_{i^{\dagger}}^{\sigma}(\tilde{b}_{i^{\dagger}})$ is far from all other players' beliefs in $b^z$:
    \[\|\tilde{\pi}_{i^{\dagger}} - b^z_{ji^{\dagger}}\|_2 > \tau, \quad \forall j \in I\setminus \{i^{\dagger}\}.\]
We define $\tilde{z}=(\tilde{b}_{i^{\dagger}}, b^z_{-i^{\dagger}}, \tilde{\pi}_{i^{\dagger}} ,\pi^z_{-i^{\dagger}})$. We compute the resistance of state transition path: $z \to \tilde{z} \to w$.

\paragraph{Transition step $z \to \tilde{z}$.} Player $i^{\dagger}$ updates their belief to $\tilde{b}_{i^{\dagger}}$, and all other players $j \neq i^{\dagger}$ keep their belief $b^z_j$. Since $z \in \Zdagger$ is a belief consistent state, we have $\|b^z_j - \pi^z_{-j}\|_2 \leq \tau$ for all $j \in I$. From \eqref{eq:edge_resistence}, we have $r_{i^{\dagger}, z\tilde{z}} = f_{i^{\dagger}}(\U{i^{\dagger}}{\pi_{i^{\dagger}},b_{i^{\dagger}}})$ and $r_{j, z\tilde{z}} = 0$ for all $j \neq i^{\dagger}$, and the edge resistance $r_{z\tilde{z}}$ between $z$ and $\tilde{z}$ is 
\[r_{z\tilde{z}} = \sum_{j \in I}r_{j, z\tilde{z}} = f_{i^{\dagger}}(\U{i^{\dagger}}{\pi^z_{i^{\dagger}},b^z_{i^{\dagger}}}).\] 

\paragraph{Transition step $\tilde{z} \to w$.} Recall that $\tilde{z}=(\tilde{b}_{i^{\dagger}}, b^z_{-i^{\dagger}}, \tilde{\pi}_{i^{\dagger}} = \Br^{\sigma}_{i^{\dagger}}(\tilde{b}_{i^{\dagger}}) ,\pi^z_{-i^{\dagger}})$, and $\|\tilde{b}_{i^{\dagger}} - \pi^z_{-i^{\dagger}}\|_2 > \tau$. Hence from \eqref{eq:edge_resistence}, we have \[r_{i^{\dagger}, \tilde{z}w} = 0.\]
Since for all $j \neq i^{\dagger}$, we have $\|\tilde{\pi}_{i^{\dagger}} - b^z_{ji^{\dagger}}\|_2 > \tau$. Then,
\[\|\pi^{\tilde{z}}_{-j} - b^{\tilde{z}}_{j}\|_2 \geq\|\tilde{\pi}_{i^{\dagger}} - b^z_{ji^{\dagger}}\|_2 > \tau.\]
From \eqref{eq:edge_resistence}, we have \[r_{j, \tilde{z}w} = 0, \quad \forall j \in I\setminus \{i^{\dagger}\}.\] Hence, the edge resistance $r_{\tilde{z}w}$ between $\tilde{z}$ and $w$ is 
\[r_{\tilde{z}w} = \sum_{j \in I}r_{j, \tilde{z}w} = 0.\] 
The path resistance of path $\rho_{z \to w} = z \to \tilde{z} \to w$ is
\[R(\rho_{z \to w}) = r_{z\tilde{z}} + r_{\tilde{z}w} = f_{i^{\dagger}}(\U{i^{\dagger}}{\pi^z_{i^{\dagger}},b^z_{i^{\dagger}}}).\]
Since the edge weight $\hat{r}_{zw}$ is defined as the smallest resistance over all possible paths between $z$ and $w$, we have 
\begin{align}\label{eq:one_direction}\hat{r}_{zw} \leq R(\rho_{z \to w}) = f_{i^{\dagger}}(\U{i^{\dagger}}{\pi^z_{i^{\dagger}},b^z_{i^{\dagger}}} = \min_{j \in I} f(\U{j}{\pi^z_i, b^z_i}).\end{align}

Next, we show that the edge weight $\hat{r}_{zw}$ is at least $\min_{i \in I} f_i(\U{i}{\pi^z_i, b^z_i})$. We argue that any transition path that starts at $z = (b^z, \pi^z) \in \Zdagger$ and ends at state $w \in \Zdagger$ has resistance larger than or equal to $\min_{i \in I} f_i(\U{i}{\pi^z_i, b^z_i})$. Consider any transition path from $z$ to $w$. Since $z \neq w$, there must exists $z' = (b', \pi') \neq z$ such that $z \to z'$ is on the path. 
Since $z' \neq z$, we have $b^z \neq b'$. Hence, there must exist at least one player $i \in I$ who changes their belief, i.e. $b^z_i \neq b'_i$. Since $\|b^z_i - \pi^z_{-i}\|_2 \leq \tau$, from \eqref{eq:edge_resistence}, we have \[r_{i, zz'} = f_i(\U{i}{\pi^z_i,b^z_i}).\]
Since $r_{j, zz'} \geq 0$ for all $j \neq i$, the edge resistance $r_{zz'}$ between $z$ and $z'$ satisfies 
\[r_{zz'} = \sum_{j \in I}r_{j, zz'} \geq r_{i, zz'} = f_i(\U{i}{\pi^z_i,b^z_i}).\] 
Since the resistance of any edge is nonnegative, for any path $\rho_{z \to w} = z \to z' \to \cdots \to w$ where $z' = (b', \pi')$ is the first state visited on the path other than $z$, the path resistance of $\rho_{z \to w}$ must satisfy:
\[R(\rho_{z \to w}) \geq r_{zz'} = f_i(\U{i}{\pi^z_i,b^z_i}), \quad \forall i \ \text{ such that } b^z_i \neq b'_i.\]
Since the edge weight $\hat{r}_{zw}$ is defined as the smallest resistance over all $\rho_{z \to w} \in \mathcal{P}_{z \to w}$, it satisfies:
\begin{align}\label{eq:the_other_direction}
\hat{r}_{zw} = \min_{\rho_{z \to w} \in \mathcal{P}_{z \to w}} R(\rho_{z \to w}) \geq \min_{i \in I} f_i(\U{i}{\pi^z_i,b^z_i}).\end{align}
Combing \eqref{eq:one_direction} and \eqref{eq:the_other_direction}, we have \[\hat{r}_{zw} = \min_{i \in I} f_i(\U{i}{\pi^z_i, b^z_i}), \quad \forall z = (b^z,\pi^z), w \in \Zdagger.\] 
\hfill $\square$

With the edge weight between consistent states characterized in Lemma~\ref{lemma: out edge resistance}, we are now ready to complete the proof of Theorem~\ref{theorem:main}. Specifically, we show that the consistent states that maximize the minimum transformed utility across all players also minimize stochastic potential, and thus constitute the stochastically stable set.

\begin{proof}[Proof of Theorem \ref{theorem:main}]
  Consider any $z \in \Zdagger$ and the associated $z$-tree. For any $w\neq z \in \Zdagger$, there exists exactly one edge that starts from $w$ in the $z$-tree. That is, there exists exactly one $w' \in \Zdagger$ such that $(w, w')$ is in the $z$-tree. From Lemma \ref{lemma: out edge resistance}, for any $w = (b^w, \pi^w) \neq z \in \Zdagger$, the edge weight is \[\hat{r}_{ww'} = \min_{i \in I} f(\U{i}{\pi_i^w, b_i^w}).\] 
    Notice that the edge weight $\hat{r}_{ww'}$ only depends on $w$. 
   Any $z$-tree must have the same total weight given by 
    \begin{align*}
        \sum_{\substack{w = (b^w, \pi^w) \\ w \in \Zdagger \setminus \{z\}}} \min_{i \in I} f(\U{i}{\pi_i^w, b_i^w}) = \sum_{w = (b^w,\pi^w) \in \Zdagger} \min_{i \in I} f(\U{i}{\pi_i^w, b_i^w}) - \min_{i \in I} f(\U{i}{\pi_i^z, b_i^z}).
    \end{align*} 
    Therefore, the stochastic potential of $z \in \Zdagger$ is \[\phi(z) =  \sum_{\substack{w = (b^w, \pi^w) \\ w \in \Zdagger \setminus \{z\}}}  \min_{i \in I} f_i(\U{i}{\pi_i^w, b_i^w}) = \underbrace{\sum_{w = (b^w,\pi^w) \in \Zdagger} \min_{i \in I} f(\U{i}{\pi_i^w, b_i^w})}_{\text{constant}} - \min_{i \in I} f(\U{i}{\pi_i^z, b_i^z}).\] 
From Lemma \ref{theorem: Young 1993}, we know that a state $z$ is stochastically stable if and only if $z$ is in a recurrent communication class and $z$ has the smallest stochastic potential. Therefore, a state $z^* = (b^*, \pi^*)$ is stochastically stable if and only if $\min_{i \in I} f(\U{i}{\pi_i^*, b_i^*})$ is the largest among all states $z \in \Zdagger$. That is,
\[\min_{i \in I} f(\U{i}{\pi_i^*, b_i^*}) = \max_{z =(b^z,\pi^z) \in \Zdagger} \min_{i \in I} f(\U{i}{\pi_i^z, b_i^z}).\]

\end{proof}

\section{Concluding Remarks}
This work develops a new learning dynamics that integrates hypothesis testing and utility-sensitive exploration, and analyzes how players adapt in general strategic environments. Beyond proving that the learning dynamics lead to playing approximate Nash equilibria for majority of the time in general normal form games, we show that the dynamics further select equilibria that maximize the minimum transformed utility across all players. This refinement is endogenously induced by the exploration mechanism adopted by agents in the learning process. Our analysis relies on interpreting the dynamics as a regular perturbation of a Markov process and characterizing the stochastically stable states using resistance-tree methods. Future directions include exploring equilibrium refinements induced by alternative exploration functions and applying this learning model to stochastic games and extensive-form games.
\bibliographystyle{apalike}
 \bibliography{library}
\newpage
\appendix

\newpage
\section{Proof of Proposition \ref{prop: hypothesis test}} \label{appendix: proof of hypothesis test}
Consider any $i \in I, b_i \in \Delta^M_{-i}, \pi_{-i} \in \Delta_{-i}$ and any $T$ observations $\{a_{-i}^t\}_{t=1}^{T}$ sampled from $\pi_{-i}$. We have
\[
\E[\hat\pi_{-i}(a_{-i})] = \pi_{-i}(a_{-i}), \quad \forall a_{-i} \in A_{-i}.
\]
Since each indicator $\mathbf1\{a_{-i}^t=a_{-i}\} \in \{0,1\}$, by Hoeffding’s inequality, for any $\delta>0$,
\begin{align*}
    &\quad P\Bigl(\Bigl|\hat\pi_{-i}(a_{-i})- \pi_{-i}(a_{-i})\Bigr| \geq \delta \Bigr)\\
    &= P\Bigl(\Bigl|\frac1T\sum_{t=1}^T\mathbf1\{a_{-i}^t=a_{-i}\} - \frac1T \E\Bigl[\sum_{t=1}^T\mathbf1\{a_{-i}^t=a_{-i}\}\Bigr]\Bigr| \geq \delta \Bigr)\\
    &=P\Bigl(\Bigl|\sum_{t=1}^T\mathbf1\{a_{-i}^t=a_{-i}\} - \E\Bigl[\sum_{t=1}^T\mathbf1\{a_{-i}^t=a_{-i}\}\Bigr]\Bigr| \geq \delta\cdot T \Bigr)\\
    &\leq 2\exp\Bigl( - \frac{2\delta^2\cdot T^2}{T}\Bigr)\\
    &= 2\exp\Bigl( - 2T \cdot \delta^2\Bigr).
\end{align*}
By Cauchy–Schwarz,
\begin{align*}
    \|\hat\pi_{-i}-\pi_{-i}\|_2
    &= \sqrt{\sum_{a_{-i} \in A_{-i}}(\hat\pi_{-i}(a_{-i})-\pi_{-i}(a_{-i}))^2}\le \sqrt{|A_{-i}|} \cdot \max_{a_{-i}\in A_{-i}}|\hat\pi_{-i}(a_{-i})-\pi_{-i}(a_{-i})|.
\end{align*}
Hence, for any $\delta>0$,
\begin{align}
    \Pr\bigl(\|\hat\pi_{-i}-\pi_{-i}\|_2\ge\delta\bigr)
    &\le \Pr\Bigl(\max_{a_{-i}\in A_{-i}}|\hat\pi_{-i}(a_{-i})-\pi_{-i}(a_{-i})|\ge\tfrac\delta{\sqrt{|A_{-i}|}}\Bigr) \notag\\
    &\le 2\exp \bigl(-2T\,(\tfrac\delta{\sqrt{|A_{-i}|}})^2\bigr) \notag\\
    &= 2\exp \left( -\frac{2T}{|A_{-i}|}\cdot \delta^2\right). \label{eq: empirical tail bound}
\end{align}

\paragraph{Type-I error analysis.}
Suppose that \(\|\pi_{-i}-b_i\| \leq \tau\). 
Let 
\[\deltaT = \sqrt{\frac{|A_{-i}|}{2T}\cdot\ln\bigl(\tfrac{2}{\alpha}\bigr)}.\]
If $H_0$ is rejected: 
 \[
  \|\hat{\pi}_{-i}-b_i\|_2  > \tau + \deltaT,
\]
then from the triangle inequality, we have
\begin{align*}
\|\hat{\pi}_{-i} - \pi_{-i}\|_2 
&\geq \|\hat{\pi}_{-i} - b_i\|_2 - \|b_i - \pi_{-i}\|_2
> \left(\tau + \deltaT \right) - \tau
= \deltaT.
\end{align*}
Therefore, 
\begin{align*}
\Pr\left(\|\hat{\pi}_{-i}-b_i\|_2 >\tau+\deltaT \right)\leq \Pr\left(\|\hat{\pi}_{-i}-\pi_{-i}\|_2>\deltaT \right)
\stackrel{\eqref{eq: empirical tail bound}}{\leq} 2\exp\left(- \frac{2T}{|A_{-i}|} \cdot \Bigl( \deltaT \Bigr)^2\right)
= \alpha.
\end{align*}

\paragraph{Type-II error analysis}
Suppose that $d:=\|\pi_{-i}-b_i\|_2 > \tau$. 
When player $i$ fails to reject $H_0$, we have
\[\|\hat{\pi}_{-i}-b_i\|_2  \leq \tau + \deltaT.\]
From the triangle inequality, we have
\[\|\hat{\pi}_{-i}-\pi_{-i}\|_2 \geq \|\pi_{-i}-b_i\|_2 - \|\hat{\pi}_{-i}-b_i\|_2 \geq d - (\tau+\deltaT).\]
Therefore,
\begin{align*}
\Pr\left(\|\hat{\pi}_{-i}-b_i\|_2 \leq \tau+\deltaT \right) \leq \Pr\left(\|\hat{\pi}_{-i}-\pi_{-i}\|_2  \geq d - (\tau+\deltaT)\right)\stackrel{\eqref{eq: empirical tail bound}}{\leq} 2\exp\left(-\frac{2T}{|A_{-i}|}\cdot \left( d-\tau-\deltaT \right)^2\right).
\end{align*}
Suppose 
\begin{align}
    T \geq T(\alpha) &= \max_{\substack{i \in I,\; b_i \in \mathcal{B}_i,\\ \pi_{-i}\in \im(\Br^{\sigma}_i),\\\|\pi_{-i}-b_i\|_2 > \tau}}\frac{2|A_{-i}|}{(\|\pi_{-i} - b_i\|_2-\tau)^2}\,\ln\bigl(\frac{2}{{\alpha}}\bigr) \geq \frac{2|A_{-i}|}{(d-\tau)^2}\,\ln\bigl(\frac{2}{{\alpha}}\bigr). \label{eq: T bar alpha}
\end{align}
Then, 
\begin{align}
    \delta_{T({\alpha})} = \sqrt{\frac{|A_{-i}|}{2T({\alpha})}\cdot\ln\bigl(\tfrac{2}{\alpha}\bigr)} \leq \frac{d-\tau}{2}. \label{eq: delta T bar alpha}
\end{align}
Since function $\exp(-x)$ is decreasing in $x$,
\begin{align*}
&\quad\; \Pr\left(\|\hat{\pi}_{-i}-b_i\|_2 \leq \tau+\delta_{T(\alpha)} \right)  \leq 2\exp\left( -\frac{2T({\alpha})}{|A_{-i}|}\cdot \left( d-\tau-\delta_{T({\alpha})}\right)^2\right) \\
&\stackrel{\eqref{eq: delta T bar alpha}}{\leq} 2\exp \left( -\frac{2T({\alpha})}{|A_{-i}|}\cdot  \left( \frac{d-\tau}{2}\right)^2\right)\stackrel{\eqref{eq: T bar alpha}}{\leq} {\alpha}.
\end{align*}
\hfill $\square$

\section{Proof of Proposition \ref{lemma: sigma, tau, M}} \label{appendix: large M}

\begin{lemma}[\cite{gao2017properties}]\label{lemma: gao}
For any $\sigma > 0$, $i \in I$, the smooth response function with temperature parameter $\sigma$ is Lipschitz in $U_i(\cdot, b_i) \in \mathbb{R}^{|A_i|}$ with Lipschitz parameter $1/\sigma$. That is,  for any $\sigma > 0$, $i \in I$, $b_i, b'_i \in \Delta_{-i}$,
\begin{align}\label{eq: logit lipschitz in U}
    \|\Br^{\sigma}_i(b_i) - \Br^{\sigma}_i(b'_i)\|_2 \leq \frac{1}{\sigma} \|U_i(\cdot,b_i) - U_i(\cdot, b'_i)\|_2.
\end{align}
\end{lemma}
\begin{proof}
    First, for any \(\sigma > 0\), define \(\pi^* \in \Delta\) as a fixed point of the smooth best response mapping. That is,
\begin{align}\label{eq:fixed_point}
    \pi^*_i(a_i) = \frac{\exp\left( \frac{1}{\sigma} U_i(a_i, \pi^*_{-i}) \right)}{\sum_{a_i' \in A_i} \exp\left( \frac{1}{\sigma} U_i(a_i', \pi^*_{-i}) \right)}, \quad \forall a_i \in A_i,\ \forall i \in I.
\end{align}
Such a fixed point \(\pi^*\) exists because the smooth best response function on the right-hand side of \eqref{eq:fixed_point} is continuous in \(\pi\), and the joint strategy space \(\Delta\) is a nonempty, compact, and convex subset of a Euclidean space. Existence then follows from Brouwer’s fixed point theorem.

We now consider a belief vector in the discretized belief space $\bdagger= (b_i^{\dagger})_{i \in I} \in \mathcal{B}= \Delta^{M}$ such that each $\bdagger_i$ has the minimum $\ell_2$ distance with the strategy profile $\pi_{-i}^*$. We also define the strategy profile $\pi^{\dagger} = \Br^{\sigma}(\bdagger)$. 

\paragraph{Part 1:} We show that $(\bdagger, \pidagger)$ satisfies the following conditions under Assumption \ref{assump: sigma, tau, M}:
\begin{subequations}
\begin{align}
    &\|\pi^*_{-i} - \bdagger_i\|_2 \leq \tau, \quad \forall i \in I,\label{eq:pi_b}\\
    &\|\Br^{\sigma}_i(\pi^*_{-i}) - \Br^{\sigma}_i(\bdagger_i)\|_2 \leq \tau, \quad \forall i \in I,\label{eq:br_b}\\
    &  \|\bdagger_i - \pidagger_{-i}\|_2 \leq \tau, \quad \forall i \in I. \label{eq:tau_consistent}
\end{align}
\end{subequations}
In particular, \eqref{eq:tau_consistent} concludes that $(b^{\dagger}, \pi^{\dagger})$ is consistent. 
\medskip 

\noindent\emph{Proof of Part 1:} To prove \eqref{eq:pi_b}, we note that since $\bdagger$ is the belief in the discretized set $\mathcal{B}$ that is closest to $\pi^{*}$, we must have 
\[|\pi^*_j(a_j) - b_{ij}^{\dagger}(a_j)| \leq \frac{1}{M}, \quad \forall a_j \in A_j, \quad \forall i, j \in I.\]
Hence, the $\ell_1$ distance between $b^{\dagger}_{ij}$ and $\pi^*_j$ can be bounded as follows:
\begin{align} \label{eq: bound on b_ij}
    \|\pi^*_j - b^{\dagger}_{ij}\|_1 = \sum_{a_j \in A_j} |\pi^*_j(a_j) - b^{\dagger}_{ij}(a_j)|
    \leq \sum_{a_j \in A_j} \frac{1}{M}
    \leq \frac{\max_{i' \in I} |A_{i'}|}{M}, \quad \forall i, j \in I.
\end{align}
Moreover, 
\begin{align}
    \|\pi^*_{-i} - b^{\dagger}_i\|_2  &\leq \|\pi^*_{-i} - b^{\dagger}_i\|_1  = \sum_{j \neq i} \|\pi^*_{j} - b^{\dagger}_{ij}\|_1 \stackrel{\eqref{eq: bound on b_ij}}{\leq}  \frac{\max_{i' \in I} |A_{i'}| \cdot |I|}{M} \notag\\
    &\stackrel{\eqref{eq: granularity M}}{\leq} \frac{\tau}{\left( 1 + \sqrt{\max_{i' \in I} |A_{i'}|} \cdot \Lu \cdot \La \cdot |I| /\sigma \right)} \label{pi_difference}\\
    &\leq \tau \notag
\end{align}

To prove \eqref{eq:br_b}, we first show that
\begin{align}
&\quad \; \|U_i(\cdot,\pi_{-i}^*) - U_i(\cdot, b^{\dagger}_i)\|_2 \notag\\
&= \left( \sum_{a_i \in A_i} \left[ \sum_{a_{-i} \in A_{-i}} (\pi_{-i}^*(a_{-i}) - b^{\dagger}_i(a_{-i})) u_i(a_i, a_{-i}) \right]^2 \right)^{1/2} \notag \\
&\leq \left( \sum_{a_i \in A_i} \left[ \|\pi_{-i}^* - b^{\dagger}_i\|_2 \cdot \|u_i(a_i, \cdot)\|_2 \right]^2 \right)^{1/2} \tag{Cauchy–Schwarz} \\
&= \|\pi_{-i}^* - b^{\dagger}_i\|_2 \cdot \left( \sum_{a_i \in A_i} \|u_i(a_i, \cdot)\|_2^2 \right)^{1/2} \leq \|\pi_{-i}^* - b^{\dagger}_i\|_2 \cdot \left( \sum_{a_i \in A_i} \sum_{a_{-i} \in A_{-i}} u_i(a_i, a_{-i})^2 \right)^{1/2} \notag \\
&\leq \|\pi_{-i}^* - b^{\dagger}_i\|_2 \cdot \left( |A_i| \cdot |A_{-i}| \cdot \Lu^2 \right)^{1/2} 
= \Lu \cdot \sqrt{|A_i| \cdot |A_{-i}|} \cdot \|\pi_{-i}^* - b^{\dagger}_i\|_2 \notag\\
&\leq \Lu \cdot \La \cdot \|\pi_{-i}^* - b^{\dagger}_i\|_2, \quad \forall i \in I. \label{eq: U over b}
\end{align}
Then, following Lemma \ref{lemma: gao}, we have
\begin{align}
    \|\Br^{\sigma}_i(\pi_{-i}^*) - \Br^{\sigma}_i(b^{\dagger}_i)\|_2 &\stackrel{\eqref{eq: logit lipschitz in U}}{\leq} \frac{1}{\sigma} \|U_i(\cdot, \pi_{-i}^*) - U_i(\cdot, b^{\dagger})\|_2 \notag\\ 
    & \stackrel{\eqref{eq: U over b}}{\leq} \frac{\Lu \cdot \La}{\sigma} \|\pi_{-i}^* - b^{\dagger}_i\|_2\notag\\
    & \stackrel{\eqref{pi_difference}}{\leq} \tau. \label{eq: logit lipschitz in b}
\end{align}

We next show that \eqref{eq:tau_consistent} holds: 
\begin{align*}
    &\quad\; \|b^{\dagger}_i - \pi_{-i}^{\dagger}\|_2 \leq \|b^{\dagger}_i - \pi_{-i}^{\dagger}\|_1 \leq \sum_{j \neq i} \|b^{\dagger}_{ij} - \pi^{\dagger}_j\|_1\\
    &= \sum_{j \neq i} \|b^{\dagger}_{ij} - \pi^*_{j} + \pi^*_{j} -\Br^{\sigma}_j(\pi_{-j}^*) +  \Br^{\sigma}_j(\pi_{-j}^*) - \pi^{\dagger}_j\|_1\\
    &\leq \sum_{j \neq i} \left( \|b^{\dagger}_{ij} - \pi^*_{j} \|_1 + \|\pi^*_{j} -\Br^{\sigma}_j(\pi_{-j}^*)\|_1 + \|\Br^{\sigma}_j(\pi_{-j}^*) - \pi^{\dagger}_j\|_1 \right)\\
    &\stackrel{}{\leq} \sum_{j \neq i} \left( \|b^{\dagger}_{ij} - \pi^*_{j} \|_1 + \|\Br^{\sigma}_j(\pi_{-j}^*) - \pi^{\dagger}_j\|_1 \right)\\
    &\leq \sum_{j \neq i} \left( \|b^{\dagger}_{ij} - \pi^*_{j} \|_1 + \sqrt{|A_j|} \|\Br^{\sigma}_j(\pi_{-j}^*) - \Br^{\sigma}_{j}(b^{\dagger}_{j})\|_2 \right)\\
    &\mathop{\leq}\limits^{\eqref{eq: bound on b_ij}}_{\eqref{eq: logit lipschitz in b}} \sum_{j \neq i} \left( \frac{\max_{i' \in I} |A_{i'}|}{M} + \sqrt{\max_{i' \in I} |A_{i'}|} \cdot \frac{\Lu \cdot \La}{\sigma} \cdot \frac{\max_{i' \in I} |A_{i'}| \cdot |I|}{M} \right)\\
    &\leq |I| \cdot \max_{i' \in I} |A_{i'}| \cdot \left( 1 + \sqrt{\max_{i' \in I} |A_{i'}|} \cdot \frac{\Lu \cdot \La}{\sigma} \cdot |I| \right) \cdot \frac{1}{M}\\
    &\stackrel{\eqref{eq: granularity M}}{\leq} \tau.
\end{align*}

\paragraph{Part 2:} We further show that $\pi^{\dagger}$ is an $\epsilon$-Nash equilibrium. 

\medskip 
\noindent\emph{Proof of Part 2:} For any $i \in I$, we note that
\begin{align}
    \|\pi^{\dagger}_i - \Br^{\sigma}_i(\pi^{\dagger}_{-i})\|_2 &= \|\Br^{\sigma}_i(b^{\dagger}_i) - \Br^{\sigma}_i(\pi^{\dagger}_{-i})\|_2 \notag\\
    &\stackrel{\eqref{eq: logit lipschitz in b}}{\leq}\frac{\Lu \cdot \La}{\sigma}\cdot \|b^{\dagger}_i -\pi^{\dagger}_{-i}\|_2 \notag\\
    &\stackrel{\eqref{eq:tau_consistent}}{\leq} \frac{\Lu \cdot \La \cdot \tau}{\sigma} \notag\\
    &\stackrel{\eqref{eq: sigma and tau}}{\leq} \frac{\epsilon}{2 \sqrt{|A|} \cdot \Lu}. \label{eq: bound on pi}
\end{align}
We use the triangle inequality to bound the difference in expected utilities. For any $i \in I$,
\begin{align}
 &|U_i(\pi^{\dagger}_i, \pi^{\dagger}_{-i}) - \max_{\pi'_i \in \Delta_i} U_i(\pi'_i, \pi^{\dagger}_{-i})|\notag\\
 \leq &|U_i(\pi^{\dagger}_i, \pi^{\dagger}_{-i}) - U_i(\Br^{\sigma}_i(\pi^{\dagger}_{-i}), \pi^{\dagger}_{-i})| + |U_i(\Br^{\sigma}_i(\pi^{\dagger}_{-i}), \pi^{\dagger}_{-i}) - \max_{\pi'_i \in \Delta_i} U_i(\pi'_i, \pi^{\dagger}_{-i})|. \label{eq:two_parts}
    \end{align}
For any $i \in I$, $\pi, \pi'\in \Delta_i$, we note that
\begin{align} 
    |\U{i}{\pi} - \U{i}{\pi'}| &= \left|\sum_{a \in A} (\pi(a) - \pi'(a)) u_i(a)\right| \leq \|\pi - \pi'\|_2 \cdot \|u_i\|_2 \notag\\
    &\leq \|\pi - \pi'\|_2 \cdot \sqrt{|A|} \max_{i \in I, a \in A} u_i(a) . \notag
\end{align}
Therefore, 
\begin{align}\label{eq: Utility lipschitz in pi}
    |U_i(\pi^{\dagger}_i, \pi^{\dagger}_{-i}) - U_i(\Br^{\sigma}_i(\pi^{\dagger}_{-i}), \pi^{\dagger}_{-i})|  \leq \sqrt{|A|} \cdot \Lu \cdot \left\|\pi^{\dagger} - \left(\Br^{\sigma}_i(\pi^{\dagger}_{-i}), \pi^{\dagger}_{-i})\right)\right\|_2. 
\end{align}
Following from \eqref{eq: Utility lipschitz in pi} and \eqref{eq: bound on pi}, we have 
\begin{align}\label{eq:part_1}
    |U_i(\pi^{\dagger}_i, \pi^{\dagger}_{-i}) - U_i(\Br^{\sigma}_i(\pi^{\dagger}_{-i}), \pi^{\dagger}_{-i})|  \leq  \frac{\epsilon}{2}.
\end{align}
Additionally, we know that 
\[\Br_i^{\sigma}(b_{i}^{\dagger})= \arg\max_{\pi_i \in \Delta_i}  U_i(\pi_i, b_i) + \sigma H(\pi)_i,
\]
where $H(\pi_i) = -\sum_{a_i} \pi_i(a_i) \log \pi_i(a_i)$ is the entropy function of $\pi_i$. 
Therefore, 
\[
U_i(\Br^{\sigma}_i(b^{\dagger}_i), b^{\dagger}_i) + \sigma H(\Br^{\sigma}_i(b^{\dagger}_i))\geq U_i(\pi^*_i, b^{\dagger}_i) + \sigma H(\pi^*_i) .\]
Since $\Br^{\sigma}_i(b_i) \in \Delta_i$ is a distribution over \( |A_i| \) elements, its entropy satisfies
$H(\Br^{\sigma}_i(b_i)) \leq \log(|A_i|)$ \cite{cover1999elements}.  
Since entropy is non-negative for any distribution, we have $H(\pi^*_i) \geq 0$.
Therefore, we have 
\begin{align*}
U_i(\Br^{\sigma}_i(b^{\dagger}_i), b^{\dagger}_i) + \sigma \cdot \log(|A_i|)&\geq U_i(\Br^{\sigma}_i(b^{\dagger}_i), b^{\dagger}_i) + \sigma \cdot H(\Br^{\sigma}_i(b^{\dagger}_i))\\
&\geq U_i(\pi^*_i, b^{\dagger}_i) + \sigma H(\pi^*_i)\geq \max_{\pi_i \in \Delta_i} U_i(\pi_i, b^{\dagger}_i).
\end{align*}
Rearranging, we have 
\begin{align}\label{eq:part_2}
    |\max_{\pi_i \in \Delta_i} U_i(\pi_i, b^{\dagger}_i) - U_i(\Br^{\sigma}_i(b^{\dagger}_i), b^{\dagger}_i)|\leq \sigma \cdot \log(|A_i|)\leq \sigma \cdot \log(\max_{i \in I}|A_i|)\stackrel{\eqref{eq: sigma and tau}}{\leq} \frac{\epsilon}{2}.
\end{align}
 Combining \eqref{eq:two_parts}, \eqref{eq:part_1} and \eqref{eq:part_2}, we have $|U_i(\pi^{\dagger}_i, \pi^{\dagger}_{-i}) - \max_{\pi'_i \in \Delta_i} U_i(\pi'_i, \pi^{\dagger}_{-i})| \leq \epsilon$, which concludes that $\pi^{\dagger}$ is an $\epsilon$-Nash equilibrium. \end{proof} 

\section{Proof of Lemma \ref{lemma: P_ep}} 
\label{appendix: transition matrix}

\begin{proof}
        Recall that the state space $\mathcal{Z}$ is 
    \[\mathcal{Z} := \{(b, \pi) \mid b_i \in \mathcal{B}_i, ~ \pi_i = \Br^{\sigma}_i(b_i), \quad \forall i \in I\}.\]
    Since $\mathcal{B}_i = \Delta^M_i$ is finite for all $i \in I$, the state space $\mathcal{Z}$ is finite. For any $\xi \in (0,1)$, $z= (b, \pi), z'= (b', \pi') \in \mathcal{Z}$, the transition probability satisfies
    \begin{subequations}
    \begin{align*}
        &\quad P^{\xi}_{zz'} \notag \\
        &\geq \Pr(\text{all players test,  fail to reject or reject and explore, then update to belief specified in } b') \notag\\
        &= \Pr(\text{all players test}) \cdot \Bigl(\sum_{I_r \subseteq I}  \Pr( I_r \text{ reject}, I\setminus I_r \text{ fail to reject }) \cdot \Pr(I\setminus I_r \text{ explore}) \Bigr) \\
        &\qquad \cdot \Pr( \text{all update belief to } b')\\
        &= \Pr(\text{all players test}) \cdot \Bigl(\sum_{I_r \subseteq I}  \Pr( I_r \text{ reject}, I\setminus I_r \text{ fail to reject }) \cdot \prod_{i \in I\setminus I_r} \xi^{f_i(\U{i}{\pi_i,b_i})} \Bigr) \notag\\
        &\qquad \cdot \Pr( \text{all update belief to } b')\\
        &\geq \Pr(\text{all players test}) \cdot \Bigl( \sum_{I_r \subseteq I} \Pr( I_r \text{ reject}, I\setminus I_r \text{ fail to reject }) \Bigr)\cdot \min_{I_r \subset I}\prod_{i \in I\setminus I_r} \xi^{f_i(\U{i}{\pi_i,b_i})} \\
        &\qquad \cdot \Pr( \text{all update belief to } b') \notag\\
        &= \prod_{i \in I}\ptest\cdot 1 \cdot \min_{I_r \subset I}\prod_{i \in I\setminus I_r} \xi^{f_i(\U{i}{\pi_i,b_i})} \cdot \prod_{i \in I} \psi_i(b'_i|b_i) > 0.
    \end{align*}
    \end{subequations}
   Since transition $P^{\xi}_{zz'}$ is positive for any $z, z' \in Z$, the Markov chain $P^{\xi}$ is aperiodic and irreducible, and the Markov process has a unique stationary distribution $\mu^{\xi}$.
\end{proof}

\section{Sufficient Condition for Assumption \ref{assump: bad strategy}}\label{appendix: bad strategy}

\begin{lemma}\label{lemma:simple_condition}
    Suppose 
    \[|\im(\Br^{\sigma}_i)| > |I|, \quad \forall i \in I,\]
    and 
    \[\tau <\min_{b_{-i} \in \mathcal{B}_{-i}, i \in I}~\max_{\tilde{b}_i\in \mathcal{B}_i } \min\left\{\|\tilde{b}_i - \Br_{-i}^{\sigma}(b_{-i})\|_2, ~ \min_{j \in I \setminus \{i\}}\{\|\Br^{\sigma}_i(\tilde{b}_i) - b_{ji}\|_2\}\right\}.\]
    Then Assumption \ref{assump: bad strategy} holds.
\end{lemma}
\begin{proof}
    We first show that 
    \[\min_{b_{-i} \in \mathcal{B}_{-i}, i \in I}\max_{\tilde{b}_i\in \mathcal{B}_i } \min\left\{\|\tilde{b}_i - \Br_{-i}^{\sigma}(b_{-i})\|_2, ~ \min_{j \in I \setminus \{i\}}\{\|\Br^{\sigma}_i(\tilde{b}_i) - b_{ji}\|_2\}\right\} > 0.\]
    Consider any $i \in I$, any $b_{-i} \in \mathcal{B}_{-i}$. Since $b_{ji}$ for $j \in I \setminus \{i\}$ are $|I|-1$ points in $\Delta_i$, and $\im(\Br^{\sigma}_i) \subset \Delta_i$, by Pigeonhole principle, there exists at least two points $\pi^\dagger, \pi' \in \im(\Br^{\sigma}_i)$ such that \[\pi' \neq \pidagger \neq b_{ji}, \quad \forall j \in I \setminus \{i\}.\]
    Let $b'_i, \bdagger_i \in \mathcal{B}_i$ be such that \[\Br^{\sigma}_i(b'_i) = \pi', \quad \Br^{\sigma}_i(\bdagger_i) = \pidagger.\]
    Since $\Br_{-i}^{\sigma}(b_{-i})$ is one point in $\mathcal{B}_i$, by Pigeonhole principle, at least one of $b'_i$ or $\bdagger_i$ is different from $\Br_{-i}^{\sigma}(b_{-i})$. Without loss of generality, suppose $\bdagger_i \neq \Br_{-i}^{\sigma}(b_{-i})$. Then, 
    \[\|\bdagger_i - \Br_{-i}^{\sigma}(b_{-i})\|_2 > 0.\]
    Since $\pidagger \neq b_{ji}$ for all $j \in I \setminus \{i\}$, 
    \[\min_{j \in I \setminus \{i\}}\{\|\Br^{\sigma}_i(\bdagger_i) - b_{ji}\|_2\} > 0.\]
    Therefore,
    \[\max_{\tilde{b}_i\in \mathcal{B}_i } \min\left\{\|\tilde{b}_i - \Br_{-i}^{\sigma}(b_{-i})\|_2, ~ \min_{j \in I \setminus \{i\}}\{\|\Br^{\sigma}_i(\tilde{b}_i) - b_{ji}\|_2\}\right\} > 0, \quad \forall i \in I, b_{-i} \in \mathcal{B}_{-i}.\]
    Consider any $i \in I$, any $b_{-i} \in \mathcal{B}_{-i}$. Let $\pi_{-i} = \Br^\sigma_{-i}(b_{-i})$, and
    \[\bdagger_i = \argmax_{\tilde{b}_i\in \mathcal{B}_i } \min\left\{\|\tilde{b}_i - \Br_{-i}^{\sigma}(b_{-i})\|_2, ~ \min_{j \in I \setminus \{i\}}\{\|\Br^{\sigma}_i(\tilde{b}_i) - b_{ji}\|_2\}\right\}.\]
    Since 
    \[\tau < \max_{\tilde{b}_i\in \mathcal{B}_i } \min\left\{\|\tilde{b}_i - \Br_{-i}^{\sigma}(b_{-i})\|_2, ~ \min_{j \in I \setminus \{i\}}\{\|\Br^{\sigma}_i(\tilde{b}_i) - b_{ji}\|_2\}\right\},\]
    we have
    \[\tau < \min\left\{\|\bdagger_i - \Br_{-i}^{\sigma}(b_{-i})\|_2, ~ \min_{j \in I \setminus \{i\}}\{\|\Br^{\sigma}_i(\bdagger_i) - b_{ji}\|_2\}\right\}.\]
    Therefore, $\bdagger_i$ and $\pi_i^{\dagger}= \Br_{i}^{\sigma}(\bdagger_i)$ satisfy 
    \begin{align*}
            \|\bdagger_i - \Br_{-i}^{\sigma}(b_{-i})\|_2> \tau, \text{ and }  \|\pi^{\dagger}_i -  b_{ji}\|_2 > \tau, \quad \forall j \neq i.
    \end{align*}
    Hence, Assumption \ref{assump: bad strategy} is satisfied. 
\end{proof}

\section{Proof of Lemma \ref{lemma: regular perturbation}}\label{apx:regular_perturbation}
\begin{enumerate}[label=(\arabic*), itemsep=1em]
    \item 
For any $z = (b, \pi) \in \mathcal{Z}$, $z' = (b', \pi
) \in \mathcal{Z}$, let $I_c$ be the set of players whose beliefs are consistent in $z$, and $I_{nc}$ be the set of players whose beliefs are not consistent in $z$:
\begin{align*}
    I_c &= \{i \in I \mid \|b_i - \pi_{-i}\|_2 \leq \tau\},\\
    I_{nc} &= \{i \in I \mid \|b_i - \pi_{-i}\|_2 > \tau\}.
\end{align*}
Let $I_d$ be the set of players whose belief changes from $z$ to $z'$, and $I_{nd}$ be set of players whose belief does not change from $z$ to $z'$
\begin{align*}
    I_d = \{i \in I \mid b_i \neq b'_i\},\\
    I_{nd} = \{i \in I \mid b_i = b'_i\}.
\end{align*}
Let $(a^t)_{t = 1}^T$ be sequence of realized action profiles in one episode. From Proposition \eqref{prop: hypothesis test},
\begin{align*}
    &\Pr\left((a^t)_{t = 1}^T \in R_i(\xi) \right) \leq \xi^{\bar{u}}, \quad \forall i \in I_c, \quad 
    \Pr\left((a^t)_{t = 1}^T \not \in R_i(\xi) \right) \leq \xi^{\bar{u}}, \quad \forall i \in I_{nc}. 
\end{align*} 
We first show that as $\xi$ goes to zero, all players with inconsistent beliefs reject their null hypothesis with probability $1$, and all players with $\tau$-consistent beliefs fail to reject their null hypothesis with probability $1$ (if they choose to conduct a test):
\begin{align}
    1&\geq \Pr( i \in I_c \text{ all fail to reject}, ~i \in I_{nc} \text{ all reject})\notag\\
    &= \Pr\left((a^t)_{t = 1}^T \not \in \bigcup_{i \in I_c} R_i(\xi), (a^t)_{t = 1}^T \in \bigcap_{i \in I_{nc}} R_i(\xi) \right)\notag\\
    &= 1 -  \Pr\left((a^t)_{t = 1}^T \in \bigcup_{i \in I_c} R_i(\xi) \text{ or } (a^t)_{t = 1}^T \not \in \bigcap_{i \in I_{nc}} R_i(\xi) \right)\notag\\
    &\geq 1- \sum_{i \in I_c} \Pr\left((a^t)_{t = 1}^T  \in R_i(\xi) \right) -  \sum_{i \in I_{nc}} \Pr\left((a^t)_{t = 1}^T \not \in R_i(\xi)\right)\notag\\
    &\geq 1- |I| \cdot \xi^{\bar{u}}. \label{eq: prob of test result}
\end{align}
Therefore, 
\begin{align}
    \lim_{\xi \to 0} \Pr( i \in I_c \text{ all fail to reject}, i \in I_{nc} \text{ all reject}) &= 1, \label{eq: correct test prob 1}\\
    \lim_{\xi \to 0} \Pr( \neg (i \in I_c \text{ all fail to reject}, i \in I_{nc} \text{ all reject})) &= 0. \label{eq: incorrect test prob 0}
\end{align}
For any $z, z' \in Z$, the transition probability $P^0_{zz'}$ from $z$ to $z'$ as $\xi \to 0$ satisfies:
\begin{align}
    & \quad P^0_{zz'} = \lim_{\xi \to 0}\Pr(z \to z') \notag\\
    &= \lim_{\xi \to 0} \Pr(z \to z' \text{ and } (i \in I_c \text{ all fail to reject}, i \in I_{nc} \text{ all reject})) \notag\\
    &\quad + \lim_{\xi \to 0} \Pr(z \to z' \text{ and } \neg (i \in I_c \text{ all fail to reject}, i \in I_{nc} \text{ all reject})) \notag\\
    &\stackrel{\eqref{eq: incorrect test prob 0}}{=}\lim_{\xi \to 0} \Pr(z \to z' \text{ and } (i \in I_c \text{ fail to reject}, i \in I_{nc} \text{ all reject})) \notag\\
    &= \lim_{\xi \to 0} \Pr(z \to z' \mid (i \in I_c \text{ fail to reject}, i \in I_{nc} \text{ all reject})) \cdot \Pr(i \in I_c \text{ fail to reject}, i \in I_{nc} \text{ all reject}) \notag\\
    &\stackrel{\eqref{eq: correct test prob 1}}{=} \lim_{\xi \to 0} \Pr(z \to z' \mid (i \in I_c \text{ fail to reject}, i \in I_{nc} \text{ all reject})) \notag.
    \end{align}
Conditioned on the event that all players \(i \in I_c\) fail to reject the null hypothesis and all players \(i \in I_{nc}\) reject the null hypothesis, the probability of transitioning from state \(z\) to state \(z'\) depends only on: (i) the exploration probabilities of players in \(I_c\), and (ii) the probabilities of resampling the specific beliefs in \(z'\) for those players who resample. These probabilities are independent across players. Therefore, 
\begin{align}    &\lim_{\xi \to 0} \Pr(z \to z' \mid (i \in I_c \text{ fail to reject}, i \in I_{nc} \text{ all reject}))\notag\\
=& \lim_{\xi \to 0} \prod_{i \in I} \Pr(z_i \to z'_i \mid (i \in I_c \text{ fail to reject}, i \in I_{nc} \text{ all reject})) \notag\\
    =& \prod_{i \in I} \lim_{\xi \to 0} \Pr(z_i \to z'_i \mid (i \in I_c \text{ fail to reject}, i \in I_{nc} \text{ all reject})). \notag
\end{align}

Analogous to the above argument, the the transition probability $P^0_{i, zz'}$ can be written as follows: 
\begin{align*} 
    P^0_{i, zz'}= \lim_{\xi \to 0} \Pr(z_i \to z'_i \mid (i \in I_c \text{ fail to reject}, i \in I_{nc} \text{ all reject})).
\end{align*}
Therefore, \begin{align*}
    P^0_{zz'} &=\prod_{i \in I} P^{0}_{i, zz'},
\end{align*}
We next compute the transition probability $P^{0}_{i, zz'}$ for each player $i$. As $\xi\to 0$, the state transition probability of player $i \in I_c$ is derived from the following event in sequence: 
\begin{itemize}
    \item Player $i$ starts a test with probability $\gamma_i$; 
    \item After a test is started, player $i$ fails to reject the null hypothesis with probability $1$ following from \eqref{eq: correct test prob 1};
    \item Player $i$ explores with probability 0 since $\lim_{\xi \to 0} \xi^{f_i(\U{i}{\pi_i,b_i})}=0$.
\end{itemize}
Therefore, 
\begin{align*}
    P^0_{i, zz'} 
    &= \begin{cases}
    1, &\text{ if }\, \|b_i - \pi_{-i}\|_2 \leq \tau,\, b_i = b'_i, \\
    0, &\text{ if }\, \|b_i - \pi_{-i}\|_2 \leq \tau,\, b_i \neq b'_i.\\
    \end{cases}
\end{align*}
Similarly, 
\begin{itemize}
    \item Each player $i \in I_{nc}$ starts a test with probability $\gamma_i$; 
    \item If a test is started, player $i$ rejects the null hypothesis with probability $1$ following \eqref{eq: correct test prob 1} and chooses a new belief according to $\psi_i(\cdot|b_i)$. 
\end{itemize}
That is, 
\begin{align*}
    P^0_{i, zz'} 
    &= \begin{cases}
    (1-\pc) + \pc \cdot \psi_i(b'_i |b_i), &\text{ if } \|b_i - \pi_{-i}\|_2 > \tau,\, b_i = b'_i, \\
    \pc \cdot \psi_i(b'_i |b_i), &\text{ if } \|b_i - \pi_{-i}\|_2 > \tau,\, b_i \neq b'_i.
    \end{cases}
\end{align*}
Summarizing,  we have
\begin{align*} 
    P^{0}_{i, zz'} &= 
\begin{cases}
    1, \quad &\text{ if }\, \|b_i - \pi_{-i}\|_2 \leq \tau,\, b_i = b'_i, \\
    0, \quad &\text{ if }\, \|b_i - \pi_{-i}\|_2 \leq \tau,\, b_i \neq b'_i, \\
    (1-\pc) + \pc \cdot \psi_i(b'_i |b_i), \quad &\text{ if } \|b_i - \pi_{-i}\|_2 > \tau,\, b_i = b'_i, \\
    \pc  \cdot \psi_i(b'_i |b_i), \quad &\text{ if } \|b_i - \pi_{-i}\|_2 > \tau,\, b_i \neq b'_i.
\end{cases}
\end{align*}

\item Condition (i) follows from Lemma \ref{lemma: P_ep}, and condition (ii) follows by the definition. 
To verify condition (iii), we need to show that $0 < \lim_{\xi \to 0} P^{\xi}_{zz'}\cdot \xi^{-r_{zz'}} < \infty$, where $r_{zz'}$ is defined in \eqref{eq:rzz} and can be rewritten as follows: 
\[r_{zz'} = \sum_{i \in I_d \cap I_c} f_i(\U{i}{\pi_i,b_i}) \geq 0, \quad \forall z = (b,\pi), z' \in \mathcal{Z}.\]
We first compute the lower bound of $\lim_{\xi \to 0} P^{\xi}_{zz'}\cdot \xi^{-r_{zz'}}$. We consider a particular event that leads the state transition from $z$ to $z'$:
\begin{itemize}
    \item All players in $I_{nd}$ (whose beliefs do not change from $z$ to $z'$) do not start hypothesis tests.
    \item All players in $I_c \cap I_d$ (whose beliefs are consistent and do change):
    \begin{enumerate}
        \item start hypothesis tests,
        \item fail to reject their null hypothesis,
        \item explore and resample new beliefs as specified in $b'$.
    \end{enumerate}
    \item All players in $I_{nc} \cap I_d$ (whose beliefs are not consistent and do change):
    \begin{enumerate}
        \item start hypothesis tests,
        \item reject their null hypothesis,
        \item update beliefs as specified in $b'$.
    \end{enumerate}
\end{itemize}
We can lower bound the transition probability $P^{\xi}_{zz'}$ by the probability of the above event: 
\begin{subequations}
\begin{align}
    P^{\xi}_{zz'} & \geq \Pr\Bigl(i \in I_{nd} \text{ does not test}, \Bigr. \notag\\
    &\qquad \quad i \in I_d\bigcap I_c \text{ tests, fails to reject, and explores to new belief } b'_i, \notag\\
    &\left. \qquad \quad  i \in I_d \bigcap I_{nc} \text{ tests, rejects, and changes to new belief } b'_i \right) \notag\\
    &= \prod_{i \in I_{nd}} \Pr\left(i\text{ does not test}\right) \cdot \prod_{i \in I_d} \Pr\left (i \text{ tests}\right) \notag\\
    &\quad \cdot \Pr\left( i \in I_d\bigcap I_c \text{ all fail to reject}, i \in I_d\bigcap I_{nc} \text{ all reject}\right) \cdot \prod_{i \in I_d\bigcap I_c} \Pr\left(i \text{ explores} \right) \notag\\
    &\quad \cdot \Pr\left( \text{all } i \in I_d \text{ choose }b'_i\right)  \label{eq: P_ep lowerbound1}\tag{*}\\
    &\geq \prod_{i \in I_{nd}} \Pr\left(i\text{ does not test}\right) \cdot \prod_{i \in I_d} \Pr\left (i \text{ tests}\right) \cdot \Pr\left( \text{all } i \in I_d \text{ choose }b'_i\right) \notag\\
    &\quad \cdot \Pr\left( i \in I_c \text{ all fail to reject}, i \in I_{nc} \text{ all reject}\right) \cdot \prod_{i \in I_d\bigcap I_c} \Pr\left(i \text{ explores} \right) \notag\\
    &\stackrel{\eqref{eq: prob of test result}}{\geq} \prod_{i \in I_{nd}} (1-\gamma_i) \cdot \prod_{i \in I_d}\gamma_i \cdot \prod_{i \in I_d} \psi_i(b'_i|b_i) \cdot \left(1- |I| \cdot \xi^{\bar{u}} \right) \cdot \prod_{i \in I_d \cap I_c}\xi^{f_i(\U{i}{\pi_i,b_i})}\notag\\
    &= \prod_{i \in I_{nd}} (1-\gamma_i) \cdot \prod_{i \in I_d}\gamma_i \cdot \prod_{i \in I_d} \psi_i(b'_i|b_i) \cdot \left(1- |I| \cdot \xi^{\bar{u}} \right) \cdot \xi^{\sum_{i \in I_d \cap I_c} f_i(\U{i}{\pi_i,b_i})}, \notag 
\end{align}
\end{subequations}
where \eqref{eq: P_ep lowerbound1} is due to the fact that players independently start a test or explore. 
Consequently, 
\begin{align*}
    &\quad \lim_{\xi \to 0} P^{\xi}_{zz'} \cdot \xi^{-r_{zz'}}\\
    & \geq \lim_{\xi \to 0} \prod_{i \in I_{nd}} (1-\gamma_i) \cdot\prod_{i \in I_d} \psi_i(b'_i|b_i) \cdot \prod_{i \in I_d}\gamma_i \cdot \left(1- |I| \cdot \xi^{\bar{u}} \right) \cdot \xi^{\sum_{i \in I_d \cap I_c} f_i(\U{i}{\pi_i,b_i})} \cdot \xi^{-\sum_{i \in I_d \cap I_c} f_i(\U{i}{\pi_i,b_i})}\\
    &= \prod_{i \in I_{nd}} (1-\gamma_i)\cdot \pc^{|I_d|} \cdot \prod_{i \in I_d} \psi_i(b'_i|b_i) > 0.
\end{align*}

Next, we compute an upper bound of $\lim_{\xi \to 0} P^{\xi}_{zz'} \cdot \xi^{-r_{zz'}}$. For any $z, z' \in Z$, the state transition probability $P^{\xi}_{zz'}$ satisfies
\begin{align}
    &\quad P^{\xi}_{zz'} \notag\\
    &= \Pr\left( \text{all } i \in I \text{ transit from $b_i$ to $b'_i$}\right) \notag\\
    &\leq \Pr\left( \text{all } i \in I_d \text{ transit from $b_i$ to $b'_i$}\right)\notag\\
    &= \Pr\left( \text{all } i \in I_d \text{ test, either reject or fail to reject and explore, then transit from $b_i$ to $b'_i$}\right)\notag \\
    &=  \Pr\left( \text{all } i \in I_d \text{ test}\right) \notag\\
    & \quad \cdot \sum_{I_r \subseteq I_d} \Bigl( \Pr\left(\text{all }i \in I_d\setminus I_r \text{ fail to reject, all } i \in I_r \text{ reject} \right)\cdot \Pr\left( \text{all } i \in I_d\setminus I_r \text{ choose to explore} \right) \Bigr)\notag\\
    &\quad \cdot \Pr\left( \text{all } i \in I_d \text{ choose }b'_i\right)\notag\\
    &= \prod_{i \in I_d}\gamma_i \cdot \prod_{i \in I_d} \psi_i(b'_i|b_i) \notag\\
    &\quad \cdot \sum_{I_r \subseteq I_d} \Bigl( \Pr\left(\text{all }i \in I_d\setminus I_r \text{ fail to reject, all } i \in I_r \text{ reject} \right)\cdot \Pr\left( \text{all } i \in I_d\setminus I_r \text{ choose to explore} \right)\Bigr).\label{eq: P_ep upperbound}
\end{align}

We consider a specific set $I_r^{\dagger} = I_d \bigcap I_{nc}$. We have $I_d\setminus I_r^{\dagger} = I_d \setminus (I_d \bigcap I_{nc}) = I_d \bigcap I_{c}$. Then, 
\begin{align*}
   &\Pr\left(\text{all }i \in I_d\setminus I_r^{\dagger} \text{ fail to reject, all } i \in I_r^{\dagger}  \text{ reject} \right)\\
   =&\Pr\left(\text{all }i \in I_d \bigcap I_c \text{ fail to reject, all } i \in I_d \bigcap I_{nc} \text{ reject} \right) \\
    \geq &\Pr\left(\text{all }i \in I_c \text{ fail to reject, all } i \in I_{nc} \text{ reject} \right) \\
    \stackrel{\eqref{eq: prob of test result}}{\geq}& 1- |I| \cdot \xi^{\bar{u}}.
\end{align*}
Therefore, 
\begin{align}
    &\quad \sum_{\substack{I_r \neq I_r^{\dagger}\\ I_r \subseteq I_d}} \Pr\left(\text{all }i \in I \setminus I_r \text{ fail to reject, all } i \in I_r \text{ reject} \right) \notag\\
    &\leq 1 - \Pr\left(\text{all }i \in I \setminus I_r^{\dagger} \text{ fail to reject, all } i \in I_r^{\dagger} \text{ reject} \right) \notag\\
    &\leq |I| \cdot \xi^{\bar{u}}. \label{eq: test result upper bound}
\end{align}
Hence,
\begin{subequations}
\begin{align}
    &\qquad P^{\xi}_{zz'} \notag\\
    &\stackrel{\eqref{eq: P_ep upperbound}}{\leq} \prod_{i \in I_d}\gamma_i \cdot \prod_{i \in I_d} \psi_i(b'_i|b_i) \cdot \notag\\
    &\quad \Bigl( \Pr\left(\text{all }i \in I_d \bigcap I_{c} \text{ fail to reject, all } i \in I_d \bigcap I_{nc} \text{ reject} \right) \cdot \Pr\left( \text{all } i \in I_d \bigcap I_c \text{ choose to explore} \right) \Bigr. \notag\\
    &\quad + \sum_{\substack{I_r \neq I_i^{\dagger}\\ I_r \subseteq I_d}} \Pr\left(\text{all }i \in I_d \setminus I_r \text{ fail to reject, all } i \in I_r \text{ reject} \right)\cdot \Pr\left( \text{all } i \in I_d \setminus I_r \text{ choose to explore} \right) \Bigl. \Bigr)\notag\\
    &\stackrel{\eqref{eq: test result upper bound}}{\leq} \prod_{i \in I_d}\gamma_i \cdot \prod_{i \in I_d} \psi_i(b'_i|b_i)\cdot \left( 1 \cdot \Pr\left( \text{all } i \in I_d \bigcap I_c \text{ choose to explore} \right) + |I| \cdot \xi^{\bar{u}} \cdot 1\right) \label{eq: P_ep upperbound cont2}\\
    &=  \prod_{i \in I_d}\gamma_i \cdot \prod_{i \in I_d} \psi_i(b'_i|b_i) \cdot \left(\prod_{i \in I_d \cap I_c}\xi^{f_i(\U{i}{\pi_i,b_i})} + |I| \cdot \xi^{\bar{u}} \right) \notag\\
    &= \prod_{i \in I_d}\gamma_i \cdot \prod_{i \in I_d} \psi_i(b'_i|b_i) \cdot \left(\xi^{\sum_{i \in I_d \cap I_c}f_i(\U{i}{\pi_i,b_i})} + |I| \cdot \xi^{\bar{u}} \right). \label{eq: P_ep upperbound cont}
\end{align}
\end{subequations}
Inequality \eqref{eq: P_ep upperbound cont2} follows from inequality \eqref{eq: test result upper bound}, and the fact that all probabilities are smaller than or equal to $1$. Then, 
\begin{subequations}
\begin{align}
    &\quad \lim_{\xi \to 0} P^{\xi}_{zz'} \cdot \xi^{-r_{zz'}} \notag\\
    &\stackrel{\eqref{eq: P_ep upperbound cont}}{\leq} \lim_{\xi \to 0} \prod_{i \in I_d}\gamma_i \cdot \prod_{i \in I_d} \psi_i(b'_i|b_i) \cdot \left(\xi^{\sum_{i \in I_d \cap I_c}f_i(\U{i}{\pi_i,b_i})} + |I| \cdot \xi^{\bar{u}} \right) \cdot \xi^{-\sum_{i \in I_d \cap I_c} f_i(\U{i}{\pi_i,b_i})}\notag\\
    &= \lim_{\xi \to 0} \prod_{i \in I_d}\gamma_i \cdot \prod_{i \in I_d} \psi_i(b'_i|b_i) \cdot \left(1 + |I| \cdot \xi^{\bar{u}-\sum_{i \in I_d \cap I_c} f_i(\U{i}{\pi_i,b_i})} \right) \notag\\
    &= \pc^{|I_d|} \cdot \prod_{i \in I_d} \psi_i(b'_i|b_i) < \infty,\notag
\end{align}
\end{subequations}
where the last equality follows from \[\bar{u}-\sum_{i \in I_d \cap I_c} f_i(\U{i}{\pi_i,b_i}) > \max_{\pi \in \Delta} \sum_{i \in I} f_i(\U{i}\pi) -\sum_{i \in I_d \cap I_c} f_i(\U{i}{\pi_i,b_i}) \geq 0.\]
Therefore, we can conclude that given $r_{zz'} = \sum_{i \in I_d \cap I_c} f_i(\U{i}{\pi_i,b_i})$, we have
\[0 < \lim_{\xi \to 0} P^{\xi}_{zz'} \cdot \xi^{-r_{zz'}} < \infty. \]

\end{enumerate}
\hfill $\square$

\end{document}